\documentclass[11pt, epsf]{article}
\usepackage{amsmath, amssymb, relsize, verbatim, graphicx, float, hyperref,titlesec,amsthm, boondox-cal, mathtools, nccmath, bm,geometry, indentfirst,enumitem}
\usepackage{color,soul}
\usepackage{setspace}
\title{Optimal Signals and Detectors Based on Correlation and Energy\footnote{This article is part of the M.Sc. thesis of the first author.}}
\author{Yossi Marciano and Neri Merhav}
\newtheorem{theorem}{Theorem}
\newtheorem{lemma}{Lemma}

\setlist[enumerate]{label*=\arabic*.}
\allowdisplaybreaks
\titleformat{\paragraph}
{\normalfont\normalsize\bfseries}{\theparagraph}{1em}{}
\titlespacing*{\paragraph}
{0pt}{3.25ex plus 1ex minus .2ex}{1.5ex plus .2ex}
\DeclareMathAlphabet{\pazocal}{OMS}{zplm}{m}{n}
\begin{document}
\maketitle
\begin{center}
The Andrew \& Erna Viterbi Faculty of Electrical and Computer Engineering\\
Technion - Israel Institute of Technology \\
Technion City, Haifa 3200003, ISRAEL \\
\end{center}
\vspace{1.5\baselineskip}
\setlength{\baselineskip}{1.5\baselineskip}

\begin{abstract}
In continuation of an earlier study, we explore a Neymann-Pearson hypothesis testing scenario where, under the null hypothesis ($\pazocal{H}_0$), the received signal is a white noise process $N_t$, which is not Gaussian in general, and under the alternative hypothesis ($\pazocal{H}_1$), the received signal comprises a deterministic transmitted signal $s_t$ corrupted by additive white noise, the sum of $N_t$ and another noise process originating from the transmitter, denoted as $Z_t$, which is not necessarily Gaussian either. Our approach focuses on detectors that are based on the correlation and energy of the received signal, which are motivated by implementation simplicity. We optimize the detector parameters to achieve the best trade-off between missed-detection and false-alarm error exponents. First, we optimize the detectors for a given signal, resulting in a non-linear relation between the signal and correlator weights to be optimized. Subsequently, we optimize the transmitted signal and the detector parameters jointly, revealing that the optimal signal is a balanced ternary signal and the correlator has at most three different coefficients, thus facilitating a computationally feasible solution.
\end{abstract}
\noindent
{\bf Index terms:} hypothesis testing, signal detection, correlation--detection, error exponent, linear programming.
\section{Introduction}
Detection of noise-corrupted signals is a topic of great interest in engineering and science, with many applications in the areas of communication
and signal processing. These include radar, sonar, light detection and ranging (LIDAR), object recognition in image and video streams, reliable communication and many more. These applications are integrated in platforms such as cell phones, cars, drones, smart watches, military equipment, etc. In many of these applications, the detector must obey low power and small area demands.

In the above-mentioned applications, the challenge of mismatch between the signal model and the detector design emerges frequently. This mismatch may stem from uncertainty in the signal model, or from restrictions on the detector design, as mentioned above. The need to place detectors in small sensors and mobile platforms, with severe constraints on power, weight, budget and size, motivates one to examine some classes of simpler detectors and find the optimal ones within these classes.

Considerable attention in the literature was given to the mismatch caused by the first reason, namely, model uncertainty, see \cite{BOR09}, \cite{GGFL98}, \cite{HLYC20}, \cite{LL19}, \cite{WLYTDY15}, just to name a few. However, the second cause described above remained relatively unattended, with the exceptions of references \cite{Neri21} and \cite{Neri22}. In both works, a Neyman-Pearson hypothesis testing problem of distinguishing between a null hypothesis, $\pazocal{H}_0$ and an alternative hypothesis $\pazocal{H}_1$, is addressed. A signal waveform $s_t$ is transmitted to detect a target by reflection (for example, in a radar system). Under $\pazocal{H}_0$ there is no target, thus the received signal $Y_t$ is a white Gaussian noise process, while under $\pazocal{H}_1$, $Y_t$ is the reflection of the transmitted signal $s_t$, corrupted by additional noise (which is not necessarily Gaussian). The sub-optimal decision rule examined in \cite{Neri21}, \cite{Neri22} is implemented by correlation detectors, namely, to compare the correlation to a certain threshold $\Theta$. If the correlation exceeds $\Theta$, the null hypothesis is rejected.

In \cite{Neri21}, the case of interest is optical detection, where the observed signal, $y(t)$, is a continuous waveform and under $\pazocal{H}_1$ models the output of a photo-diode. First, a stream of photo-electrons is generated by photons sensed by the diode in the time interval $[0,T)$. Then, the current generated by the diode is amplified by a trans-impedance amplifier. Finally, an additive white Gaussian thermal noise process corrupts the sensor's output. Throughout the scope of \cite{Neri21}, the transmitted signal model remains fixed, and the detector designer cannot optimize it. The correlation therein was obtained by calculating $\int_0^Tw(t)y(t)\mbox{d}t$, where $w(t)$ is a predefined continuous correlator waveform. In \cite{Neri22} the received signal $Y_t$ is a discrete-time signal, and the correlation is defined as $\sum_{t=1}^n w_tY_t$, where $\textbf{w}=(w_1,\dots,w_n)$ is a vector of real valued correlator coefficients. In both works, the optimal correlator configuration is found in the sense of maximizing the miss-detection (MD) error exponent subject to a given false-alarm (FA) error exponent. However, the approach in \cite{Neri22} is more general in some sense, since the statistics under $\pazocal{H}_1$ are not fully detailed as in \cite{Neri21}, and only general, reasonable assumptions are made regarding the noise processes in the problem. Also, the case where the transmitted signal is allowed to be optimized alongside the correlator weights was considered in \cite{Neri22}. An informal description of the problem is as follows: Consider the two hypotheses,
\begin{align}
\pazocal{H}_{0}\colon \qquad Y_{t} &= N_t\qquad 		t=1,\dots,n\\ 
\pazocal{H}_{1}\colon \qquad Y_{t} &= X_{t} + N_t\qquad t=1,\dots,n,
\end{align}
where ${X_t}$ is a random process decomposed as $X_{t}=s_{t}+Z_{t}$, with $s_{t}=\mathbb{E}\left [X_{t}\right ]$ being a deterministic waveform and $Z_{t}=X_{t}-s_{t}$ being an independently identically distributed (i.i.d.), zero-mean noise process with a known probability density function (PDF) $f_{\mbox{\tiny{Z}}}(\cdot)$. This noise component is regarded as signal induced noise (SIN) and it is not necessarily Gaussian. The first part of \cite{Neri22} assumes a given, predefined model of the received signal, and in the second part, the case where $\left \{w_t\right \}$ and $\left \{s_t\right \}$ is are optimized jointly is analysed. Anyhow, throughout the scope of the \cite{Neri22}, $N_t$ is assumed to be a white Gaussian noise process.

In the current study, the scope of in \cite{Neri22} is extended by relaxing the assumption that $N_t$ is Gaussian. Instead, $N_{t}$ is assumed to have a general PDF, following the same assumptions as those of $Z_t$. Similarly as in \cite{Neri21} and \cite{Neri22}, we say that a detector is asymptotically optimal (within a certain class) if for a prescribed FA error exponent, its MD error exponent is maximal.
With this new setting, the following issues were investigated:
\begin{enumerate}
\item
For a fixed, pre-defined signal, the optimal detector based on computing the correlation was derived as a (non-linear) function of the transmitted signal. 
\item
The transmitted signal and the correlator weights were jointly optimized. This yielded an optimal signal $\left \{s_t\right \}$ which is quite simple, since it has only two parameters to be optimized. The correlator weights $\left \{w_t\right \}$ turned out to be proportional to $\left \{s_t\right \}$, even though the problem is not Gaussian in general. These results improve the ones presented in \cite{Neri22} regarding joint signal-correlator optimization, since it was discovered that the optimal signal-correlator combination is much simpler than in \cite{Neri22}, with only four parameters to be optimized.
\item
Most importantly, the analysis performed in items 1 and 2 above was derived fully for the test statistic $\pazocal{T}\triangleq\sum_t \left (w_tY_t+\gamma Y_t^2\right )$ (linear combination of correlation and energy), where $\gamma\in\mathbb{R}$ was subjected to optimization. We are motivated to examine this detector since under $\pazocal{H}_1$, the observed signal has a higher energy as well as correlation with the transmitted signal and so, it is conceivable that the energy of the noisy received signal would bear information that is relevant to distinguish  between the two hypotheses. In fact, if both $N_t$ and $Z_t$ are Gaussian, then the likelihood-ratio test (LRT) statistic is indeed of the form of a linear combination of correlation and energy. Also, the complexity of $\pazocal{T}$ is not much worse than the complexity when setting $\gamma=0$, if $\pazocal{T}$ is written as $\pazocal{T}=\sum_t  Y_t\cdot\left (w_t+\gamma Y_t\right )$. For a prescribed signal $\left \{s_t\right \}$, the correlator weights $\{w_t\}$ are given as a non linear function of $\{s_t\}$, similarly to item 1. When joint signal-detector optimization was carried out, $\left \{w_t\right \}$ and $\left \{s_t\right \}$ did not sustain a linear relation in general. Interestingly, the favorable property that the set $\{w_t,s_t\}$ has a simple structure was preserved, with just three parameters to be optimized, which makes a numeric solution feasible. The full derivation for this kind of detector was not done in \cite{Neri22}.
\end{enumerate}

The outline of the paper is as follows. In Section \ref{Section2}, we provide a formal presentation of the problem. In Section \ref{Section3}, we extend the optimal correlator presented in \cite{Neri22} for non-Gaussian noises, both for a given signal and for jointly-optimal signal and detector. Finally, in Section \ref{Section4}, the detector based on combination of energy and correlation will be derived fully.
\section{Problem Formulation and Preliminaries}\label{Section2}
The hypothesis testing problem outlined in the Introduction is now elaborated upon. Each entry of the random process $\left \{N_t\right \}$ (resp. $\left \{Z_t\right \}$) is an independent copy of a zero-mean random variable (RV) $N$ (resp. $Z$) with a PDF $f_{\mbox{\tiny{N}}}(\cdot)$ (resp. $f_{\mbox{\tiny{Z}}}(\cdot)$), which is symmetric around the origin. The assumption regarding the symmetric distribution of $N$ and $Z$ is made for the sake of convenience and simplicity, and not due to a principal hurdle to tackle the non-symmetric case. Also, the symmetry of noise processes is conceivable as a realistic and reasonable assumption. The RVs $N$ and $Z$ are independent, with cumulant-generating functions (CGFs) $C_{\mbox{\tiny{N}}}(\cdot)$ and $ C_{\mbox{\tiny{Z}}}(\cdot)$, respectively
\begin{equation}\label{CGF}
C_{\mbox{\tiny{N}}}(v)\triangleq \ln\mathbb{E}\left [e^{vN}\right ];\;C_{\mbox{\tiny{Z}}}(v)\triangleq \ln\mathbb{E}\left [e^{vZ}\right ].
\end{equation}
A joint CGF of the random vectors $(V,V^2)$ (resp. $(N,N^2)$) is defined for future use as $\tilde{C}_{\mbox{\tiny{V}}}(x,y)$ and $\tilde{C}_{\mbox{\tiny{N}}}(x,y)$, respectively
\begin{equation}\label{JCGF}
\tilde{C}_{\mbox{\tiny{V}}}(x,y)\triangleq \ln\mathbb{E}\left [e^{xV+yV^2}\right ]; \; \tilde{C}_{\mbox{\tiny{N}}}(x,y)\triangleq \ln\mathbb{E}\left [e^{xN+yN^2}\right ].
\end{equation}
where $V\triangleq N+Z$.

We assume that \eqref{CGF} and \eqref{JCGF} are finite and twice differentiable, at least in a certain domain around the origin. It is noted that for an arbitrary RV $J$ with a symmetric PDF around the origin, we have $C_{\mbox{\tiny{J}}}(v)=C_{\mbox{\tiny{J}}}(-v)$ and $\tilde{C}_{\mbox{\tiny{J}}}(x,y)=\tilde{C}_{\mbox{\tiny{J}}}(-x,y)$. It is also straightforward to verify that both $\tilde{C}_{\mbox{\tiny{J}}}$ and $C_{\mbox{\tiny{J}}}$ are strictly convex functions, unless $J$ is a degenerated RV.

The observed signal is a vector, $(Y_1, Y_2,\dots, Y_n)\in\mathbb{R}^n$, the correlator weight vector, $\textbf{w}\triangleq (w_1, w_2,\dots, w_n)\in\mathbb{R}^n$, the signal transmitted is $\textbf{s}\triangleq (s_1,s_2,\dots, ,s_n)\in\mathbb{R}^n$ and we assume that the amplitudes of $\textbf{w}$ and $\textbf{s}$ are bounded. The test statistic, $\pazocal{T}\triangleq\sum_{t=1}^n\left ( w_tY_t+\gamma Y_t^2\right )$ is compared to a threshold $\Theta\triangleq n\theta$, where $\theta$ is a constant, independent of $n$, and $\gamma\ge 0$ is a coefficient subjected to optimization. If $\pazocal{T}$ exceeds $\Theta$, the null hypothesis is rejected, otherwise, it is accepted. Throughout the  paper, our goal is to find asymptotically optimal detectors in the sense of minimizing the MD probability, $P_{\mbox{\tiny{MD}}}$,
\begin{equation}
P_{\mbox{\tiny{MD}}}\triangleq \Pr\left \{ \pazocal{T}< \Theta |\pazocal{H}_1 \right \},
\end{equation}
subject to a constraint on the maximal allowable FA probability, $P_{\mbox{\tiny{FA}}}$,
\begin{equation}
P_{\mbox{\tiny{FA}}}\triangleq \Pr\left \{ \pazocal{T}\ge \Theta |\pazocal{H}_0 \right \},
\end{equation}
as $n$ tends to infinity, where $\Pr\{\cdot|\cdot \}$ is the conditional probability. More precisely, define the MD error exponent as 
\begin{equation}
E_{\mbox{\tiny{MD}}}\triangleq \lim_{n\to\infty}\left \{-\frac{1}{n}\ln\left (P_{\mbox{\tiny{MD}}}\right )\right \},
\end{equation}
and the FA error exponent,
\begin{equation}
E_{\mbox{\tiny{FA}}}\triangleq \lim_{n\to\infty}\left \{-\frac{1}{n}\ln\left (P_{\mbox{\tiny{FA}}}\right )\right \},
\end{equation}
provided that theses two limits exist. We wish to maximize the MD error exponent for a prescribed FA error exponent. With a slight abuse of notation, we will occasionally express $E_{\mbox{\tiny{FA}}}$ and $E_{\mbox{\tiny{MD}}}$ as functions of different variables, depending on the context.

In the first part of the article, we examine linear detectors, namely, $\pazocal{T}=\sum_{t=1}^n w_tY_t$ ($\gamma$=0). For a given $\textbf{s}$, we aim to find $\textbf{w}$ which maximizes $E_{\mbox{\tiny{MD}}}$ for a given $E_{\mbox{\tiny{FA}}}$. Next, we allow the signal and the detector to be optimized jointly for all signals with $P(\textbf{s})\leq P_{\mbox{\tiny{s}}}$, where
\begin{equation}\label{PowerConstraint}
P(\textbf{s})\triangleq \frac{1}{n}\sum_{t=1}^n s_t^2,
\end{equation}
and $P_{\mbox{\tiny{s}}} > 0$ is a given constant. Our goal is to find the optimal combination of $\textbf{s}$ and $\textbf{w}$ in terms of the trade-off between the FA exponent and the MD exponent. In the second part of the paper, the class of detectors is expanded to include detectors based on linear combinations of correlation and energy, namely, the energy coefficient $\gamma$ is allowed to be optimized. Similarly to the derivation of the linear detector, our goal here is to find the optimal $\textbf{w}$ per given $\textbf{s}$, and then find the optimal pair ($\textbf{w},\textbf{s}$) which obeys $P(\textbf{s})\leq P_{\mbox{\tiny{s}}}$ (per defined FA error exponent). As we derive solutions for the problems listed above, their unique properties and structure are characterized with the intention of finding optimal solutions of the simplest form. These beneficial attributes are then exploited to efficiently solve several numerical examples, which we compare to other methods.
\section{Detectors Based On Correlation}\label{Section3}
In this section, we focus on the class of correlation detectors ($\gamma=0$), in correspondence to items 1 and 2 in the Introduction. The MD error probability associated with the correlator is upper bounded by the Chernoff bound\footnote{The Chernoff bound is well-known to be exponentially tight as $n$ tends to infinity, see \cite{ChernoffBound}.}:
\begin{flalign}\label{P_MD}
P_{\mbox{\tiny{MD}}}&=\Pr\left\{\sum_{t=1}^nw_t(s_t+Z_t+N_t)<\theta n\right\}
& \nonumber \\
& \le\inf_{\lambda\ge 0}\left \{\mathbb{E}\left[\exp\left(\lambda\theta n-\lambda\sum_{t=1}^nw_t(s_t+Z_t+N_t)\right)\right]\right \}
& \nonumber \\
&=\inf_{\lambda\ge 0}\left \{\exp\left(\lambda\theta
n-\lambda\sum_{t=1}^nw_ts_t\right)\cdot\mathbb{E}\left [\exp\left(-\lambda\sum_{t=1}^nw_tZ_t\right)\right ]\cdot
\mathbb{E}\left [\exp\left(-\lambda\sum_{t=1}^nw_tN_t\right)\right ]\right \}
& \nonumber \\
&= \inf_{\lambda\ge 0}\left \{\exp\left(\lambda\theta
n-\lambda\sum_{t=1}^nw_ts_t\right)\cdot\prod_{t=1}^n
\mathbb{E}\left [\exp\left (-\lambda w_tZ_t\right )\right ]\cdot\prod_{t=1}^n
\mathbb{E}\left [\exp\left (-\lambda w_tN_t\right )\right ]\right \}
& \nonumber \\
&= \inf_{\lambda\ge 0}\left \{\exp\left(\lambda\theta
n-\lambda\sum_{t=1}^nw_ts_t\right)\cdot\prod_{t=1}^n
e^{C_{\mbox{\tiny{Z}}}\left (-\lambda w_t\right )}\cdot\prod_{t=1}^n
e^{C_{\mbox{\tiny{N}}}\left (-\lambda w_t\right )}\right \}
& \nonumber \\
&\stackrel{(a)}{=} \inf_{\lambda\ge 0}\left \{\exp\left(\lambda\theta
n-\lambda\sum_{t=1}^nw_ts_t\right)\cdot\prod_{t=1}^n
e^{C_{\mbox{\tiny{Z}}}\left (\lambda w_t\right )}\cdot\prod_{t=1}^n
e^{C_{\mbox{\tiny{N}}}\left (\lambda w_t\right )}\right \}
& \nonumber \\
&= \exp\left\{\inf_{\lambda\ge 0}\left [\lambda\theta
n-\lambda\sum_{t=1}^n w_ts_t+\sum_{t=1}^n
C_{\mbox{\tiny{Z}}}\left (\lambda w_t\right )+C_{\mbox{\tiny{N}}}\left (\lambda w_t\right )\right ]\right\}
& \nonumber \\
&= \exp\left\{\inf_{\lambda\ge 0}\left [\lambda\theta
n-\lambda\sum_{t=1}^n w_ts_t+\sum_{t=1}^n
C_{\mbox{\tiny{V}}}\left (\lambda w_t\right )\right ]\right\}
& \nonumber \\
&= \exp\left\{-\sup_{\lambda\ge 0}\left [-\lambda\theta
n+\lambda\sum_{t=1}^n w_ts_t-\sum_{t=1}^n
C_{\mbox{\tiny{V}}}\left (\lambda w_t\right )\right ]\right\},
\end{flalign}
where (a) is due to the symmetry of $C_{\mbox{\tiny{N}}}(\cdot )$ and $C_{\mbox{\tiny{Z}}}(\cdot )$ around the origin.

The MD error exponent is therefore given by
\begin{equation}\label{E_MD_lim}
E_{\mbox{\tiny{MD}}}=\lim_{n\to\infty}\left \{\sup_{\lambda\ge 0}\left [-\lambda\theta
+\frac{1}{n}\lambda\sum_{t=1}^n w_ts_t-\frac{1}{n}\sum_{t=1}^n
C_{\mbox{\tiny{V}}}\left (\lambda w_t\right )\right ]\right \},
\end{equation}
provided that the limit exists. For this limit to exist, we assume the following, similarly as in \cite{Neri22}: the pairs $\left \{\left (w_t,s_t\right )\right \}_{t=1}^n$ have an empirical joint PDF that tends to a certain asymptotic $f_{\mbox\tiny{{WS}}}(w,s)$ as $n\longrightarrow\infty$ and so, \eqref{E_MD_lim} exists. More precisely, we assume that
\begin{equation}\label{assumption2}
E_{\mbox{\tiny{MD}}}(\lambda)\triangleq\lim_{n\to\infty}\left \{-\lambda\theta
+\frac{1}{n}\lambda\sum_{t=1}^n w_ts_t-\frac{1}{n}\sum_{t=1}^n
C_{\mbox{\tiny{V}}}\left (\lambda w_t\right )\right \}=\mathbb{E}\left [\lambda(WS-\theta)-C_{\mbox{\tiny{V}}}\left (\lambda W\right )\right ],
\end{equation}
exists and the convergence in \eqref{assumption2} is uniform in $\lambda$. Here, the expectation is with respect to (w.r.t.) $f_{\mbox\tiny{{WS}}}$. The MD error exponent in \eqref{E_MD_lim} may be presented as 
\begin{equation}
E_{\mbox{\tiny{MD}}}=\sup_{\lambda\ge 0}\left \{E_{\mbox{\tiny{MD}}}(\lambda)\right \}.
\end{equation}
The FA error exponent for the correlator can be derived similarly. First, the FA error probability is upper bounded by the Chernoff bound,
\begin{flalign}\label{P_FA}
P_{\mbox{\tiny{FA}}} &=\Pr\left\{ \sum_{t=1}^{n}w_{t}N_{t}\geq\theta{n} \right\}
& \nonumber \\
&\leq\inf_{\alpha\geq{0}}\left\{ e^{-\alpha\theta{n}}\prod_{t=1}^{n}e^{C_{\mbox{\tiny{N}}}(\alpha{w_{t}})} \right\}
& \nonumber \\
&=\inf_{\alpha\geq{0}}\left\{ \exp\left[ -\alpha\theta{n} + \sum_{t=1}^{n}C_{\mbox{\tiny{N}}}(\alpha{w_{t}}) \right] \right\}
& \nonumber \\
&=\exp\left\{ -\sup_{\alpha\geq{0}}\left[ \alpha\theta n - \sum_{t=1}^{n}C_{\mbox{\tiny{N}}}(\alpha{w_{t}}) \right] \right\},
\end{flalign}
which yields a FA error exponent as follows,
\begin{equation}\label{FA_exp}
E_{\mbox{\tiny{FA}}}= \lim_{n\to\infty}\left \{\sup_{\alpha\geq{0}}\left[ \alpha\theta - \frac{1}{n}\sum_{t=1}^n C_{\mbox{\tiny{N}}}(\alpha w_t)\right]\right \},
\end{equation}
where the limit in \eqref{FA_exp} exists under a similar assumption regarding the asymptotic PDF, $f_{\mbox\tiny{{WS}}}$. Again, we assume that
\begin{equation}\label{assumption1}
E_{\mbox{\tiny{FA}}}(\alpha)\triangleq \lim_{n\to\infty}\left \{ \alpha\theta - \frac{1}{n}\sum_{t=1}^n C_{\mbox{\tiny{N}}}(\alpha w_t)\right\}=\alpha\theta-\mathbb{E}\left [C_{\mbox{\tiny{N}}}\left (\alpha W\right )\right ],
\end{equation}
exists the convergence is uniform in $\alpha$. Due to the uniform convergence of $E_{\mbox{\tiny{FA}}}(\alpha)$,  \eqref{FA_exp} is equivalent to
\begin{equation}
E_{\mbox{\tiny{FA}}}= \sup_{\alpha\geq{0}}\left\{ E_{\mbox{\tiny{FA}}}(\alpha)\right \}.
\end{equation}
\subsection{Optimal Correlator for a Given Signal}
We now address the problem of finding the asymptotically optimal correlator for a given signal, which is equivalent to finding the optimal conditional PDF, $f_{W|S}$, such that, 
\begin{equation}
\mathbb{E}\left [\lambda(WS-\theta)-C_{\mbox{\tiny{V}}}\left (\lambda W\right )\right ]= \int_{-\infty}^{\infty}f_{\mbox{\tiny{S}}}(s)\cdot\mathbb{E}\left [\lambda{sW} - C_{\mbox{\tiny{V}}}(\lambda{W})|S=s\right ]\mbox{d}s -\lambda\theta,
\end{equation}
is maximal w.r.t. $f_{W|S}$, subject to the FA error constraint, 
\begin{equation}\label{19}
E_{\mbox{\tiny{FA}}} \geq E_{\mbox{\tiny{0}}},
\end{equation}
where $E_{\mbox{\tiny{0}}}\ge 0$ is the prescribed value of the minimum required FA error exponent. The optimal density $f_{W|S}$ is characterised by Theorem \ref{Theorem1} for a fixed $\lambda$, which is eventually optimized as well.
\begin{theorem}\label{Theorem1}
Let the assumptions of Section \ref{Section2}, Section \ref{Section3} hold and let $N$ be a non-degenerate RV. Define the function $g$ as 
\begin{equation}
g(w;\alpha, \varphi, \lambda)\triangleq\dot{C}_{\mbox{\tiny{V}}}\left (\lambda w\right )+\frac{\mathbb{\varphi}\alpha}{\lambda}\dot{C}_{\mbox{\tiny{N}}}\left (\alpha w\right ),
\end{equation}
where $\dot{C}_{\mbox{\tiny{V}}}$ and $\dot{C}_{\mbox{\tiny{N}}}$ denote the derivatives of $C_{\mbox{\tiny{V}}}$ and $C_{\mbox{\tiny{N}}}$, respectively. Further assume that there exists $\varphi\ge 0$ (possibly, dependent on $\lambda$ and $\alpha$) that satisfies:
\begin{equation}\label{21a}
\mathbb{E}\left [\alpha C_{\mbox{\tiny{N}}}\left (g^{-1}\left (S;\alpha, \varphi, \lambda\right )\right )\right ]=\alpha\theta-E_{\mbox{\tiny{0}}},
\end{equation}
where $g^{-1}$ is the inverse function of $g$ (as a function $w$). If no such $\varphi$ exists, set $\varphi=0$. The optimal conditional PDF is given by
\begin{equation}
f^*_{W|S}(w|s)=\delta\left (w-g^{-1}\left (s; \alpha, \varphi, \lambda\right )\right ),
\end{equation}
where $\delta(\cdot )$ is the Dirac delta function. and 
\begin{equation}
 \alpha=\arg\max_{\alpha\ge 0}\left \{\mathbb{E}\left  [\lambda{g}^{\mbox{\tiny{-1}}}\left (S;\alpha,\varphi,\lambda\right)\cdot S-{C}_{\mbox{\tiny{V}}}\left (\lambda {g}^{\mbox{\tiny{-1}}}\left (S;\alpha,\varphi,\lambda\right)\right )\right  ]\right \}
\end{equation} 
\end{theorem}
\begin{proof}
Theorem \ref{Theorem1} is a special case of Theorem \ref{Theorem2} below, where $\gamma=0$. See the proof therein.
\end{proof}
Two remarks are now in place:
\begin{enumerate}
\item
A similar theorem to Theorem \ref{Theorem1} was derived in \cite{Neri22} (Theorem 1 therein), except it was assumed that $\left \{N_t\right \}$ was a white Gaussian noise process, and so, Theorem 1 of \cite{Neri22} is a special case of Theorem \ref{Theorem1} herein.
\item
The existence of $\varphi>0$ which complies to \eqref{21a} is reasonable, since $\varphi$ is the Lagrange multiplier of the FA error constraint \eqref{19}. An optimal solution will likely sustain the constraint in equality, implying that $\varphi>0$.
\end{enumerate}
\noindent\textbf{Example 1.}\label{Example1}
To better understand the results thus far, a numerical example is presented and solved. In the problem of interest, $\left \{s_t\right \}$ is a 4-ASK signal, $s_t\in\left \{-3a,-a,a,3a\right \}$, where $a$ is a given positive constant. Let us set: $s_t=\pm a$ along half of the time and $s_t=\pm 3a$ along the other half. Due to the symmetric distributions of $N$ and $V$ around the origin, the optimal correlator weight vector $\{w^*_t\}$ is a finite-alphabet signal as well, with the same alphabet size as $\{s_t\}$ and a similar structure, namely  $w_t\in\{-\beta,-\delta,\delta,\beta\}$, where $0<\delta<\beta$ are the correlator weights to be optimized. In this example, we assume that $Z$ is a binary symmetric random variable, taking values $\pm z_0$ for some $z_0>0$ and $N$ has a Laplace distribution 
\begin{equation}
f_{\mbox{\tiny{N}}}(v)=\frac{q}{2}\cdot e^{-q\cdot |v|},
\end{equation}
where $q>0$ is a given scale parameter. The CGF of $Z$ is given by 
\begin{align}
C_{\mbox{\tiny{Z}}}(v) = \ln\cosh(z_0v),
\end{align}
and the CGF of $N$ is given by 
\begin{equation}
C_{\mbox{\tiny{N}}}(v)=-\ln\left(1-\frac{v^2}{q^2}\right), \quad |v|<\frac{q^2}{v^2}
\end{equation}
First, we verify that the error exponents in the problem are well defined. This requires that $E_{\mbox{\tiny{FA}}}(\alpha)$ (resp. $E_{\mbox{\tiny{MD}}}(\lambda)$) defined in \eqref{assumption1} (resp. \eqref{assumption2}) converge uniformly w.r.t $\alpha$ (resp. $\lambda$). To show that the FA converges uniformly over the interval $\alpha\ge 0$, we assume without essential loss of generality that $n$ is even and consider that
\begin{flalign}\label{27}
\alpha\theta - \frac{1}{n}\sum_{t=1}^n C_{\mbox{\tiny{N}}}(\alpha w_t)\stackrel{(a)}{=}\alpha\theta - \frac{1}{2}C_{\mbox{\tiny{N}}}(\alpha \beta)-\frac{1}{2}C_{\mbox{\tiny{N}}}(\alpha \delta)=E_{\mbox{\tiny{FA}}}(\alpha),
\end{flalign}
where (a) follows since the correlator weights are matched to the signal. Since $E_{\mbox{\tiny{FA}}}(\alpha)$ is independent of $n$, the uniform convergence is trivial. Similarly, it is easy to show that the MD error exponent in this example complies with our assumptions.

Since the size of the alphabet of $\{w_t\}$ is small and its form is symmetric, we opt to optimize $\{w_t\}$ directly subject to the FA error constraint. In the setting described, the MD error exponent takes the form 
\begin{align}\label{24a}
E_{\mbox{\tiny{MD}}}(E_{\mbox{\tiny{0}}}, f_{\mbox{\tiny{S}}})=\frac{1}{2}\cdot \sup_{\lambda\geq 0}\max_{\left\{\delta,\,\beta:\; E_{\mbox{\tiny{FA}}}\geq E_{\mbox{\tiny{0}}}, \;0<\delta<\beta \right\}}\bigg\{\lambda a\delta+3\lambda a\beta - C_{\mbox{\tiny{V}}}(\lambda\delta)- C_{\mbox{\tiny{V}}}(\lambda\beta)-2\lambda\theta\bigg\},
\end{align}
where the FA error exponent can be written more compactly as
\begin{equation}\label{25a}
E_{\mbox{\tiny{FA}}}= \frac{1}{2}\cdot \sup_{\alpha\geq{0}}\left\{2\alpha\theta -C_{\mbox{\tiny{N}}}(\alpha\delta)-C_{\mbox{\tiny{N}}}(\alpha\beta) \right\}.
\end{equation}
We chose to tackle the  FA constraint by substituting it into the objective function. For a feasible solution, there exists $\alpha\ge 0$ such that the FA constraint is kept, namely
\begin{align}
\nonumber \alpha\theta-\mathbb{E}\left [C_{\mbox{\tiny{N}}}(\alpha W)\right ]\ge E_{\mbox{\tiny{0}}} \\
\theta\ge \frac{1}{\alpha}\cdot \left [\mathbb{E}\left [C_{\mbox{\tiny{N}}}(\alpha W)\right ]+E_{\mbox{\tiny{0}}}\right ]
\end{align}
The optimal $\theta$ to be considered is of the form
\begin{equation}\label{28a}
\theta = \frac{1}{2}\cdot \inf_{\alpha\ge 0}\left \{\frac{1}{\alpha}\left [2E_{\mbox{\tiny{0}}}+C_{\mbox{\tiny{N}}}(\alpha\delta)+C_{\mbox{\tiny{N}}}(\alpha\beta)\right ]\right \},
\end{equation}
since it ensures the FA error constraint is kept and minimal degradation is inflicted upon the MD error exponent, which turns out to be
\begin{align}\label{29}
E_{\mbox{\tiny{MD}}}\left (E_{\mbox{\tiny{0}}}, f_{\mbox{\tiny{S}}}\right )=\frac{1}{2}\cdot \sup_{\lambda\geq 0,\; 0<\delta<\beta,\; \alpha\ge 0}\left\{a\lambda \delta+3a\lambda\beta - C_{\mbox{\tiny{V}}}(\lambda\delta)- C_{\mbox{\tiny{V}}}(\lambda\beta)-\frac{\lambda}{\alpha}\left [2E_{\mbox{\tiny{0}}}+C_{\mbox{\tiny{N}}}(\alpha\delta)+C_{\mbox{\tiny{N}}}(\alpha\beta)\right ]\right\}.
\end{align}
Results of the optimization in \eqref{29} for $E_{\mbox{\tiny{0}}}\in[0,15]$ are presented in Figure \ref{fig:Example2ROC} for specific values of $a$, $z_0$ and $q$. In this figure the optimal correlator is compared to ``classical'' correlator, where the correlator weights are proportional to the signal levels, namely $w_t\propto s_t$ or $\beta = 3\alpha$.
\begin{figure}[H]
\centering
\includegraphics[scale=0.5]{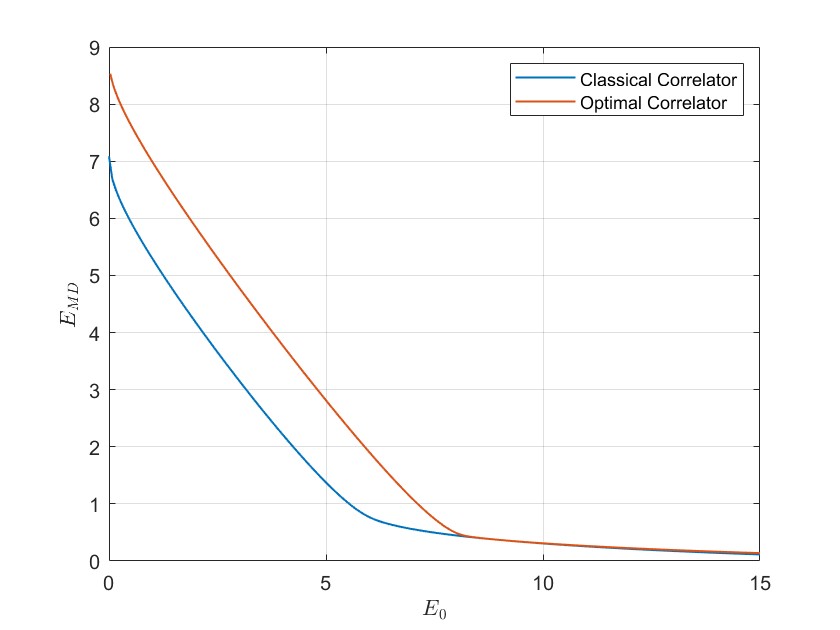}
\caption{MD error exponents as functions of $E_{\mbox{\tiny{0}}}$ for the classical correlator and the optimal correlator with parameter values $a=4$, $z_0=7$ and $q=4$.}
\label{fig:Example2ROC}
\end{figure}

We see in Figure \ref{fig:Example2ROC} that for values of $E_{\mbox{\tiny{0}}}$ smaller than approximately 7.5 the optimal correlator significantly outperforms the classical correlator, but for larger values of $E_{\mbox{\tiny{0}}}$, the optimal correlator bears very little advantage over the classical, sub-optimal one. This is because for high values of $E_{\mbox{\tiny{0}}}$, $E_{\mbox{\tiny{MD}}}$ becomes smaller, and as $E_{\mbox{\tiny{MD}}}\longrightarrow 0$, $\lambda$ in \eqref{29} tends to zero. This in turn, implies that the CGF of $V$ is evaluated close to the origin, where it does not differ significantly from a Gaussian CGF, which is a quadratic function. As for $C_{\mbox{\tiny{N}}}$ , a plausible intuition for the similar characteristics of the classical correlator and the optimal one is that the Laplacian PDF is much `closer' to the Gaussian PDF compared to a binary PDF, for example, and therefore, the mismatch between the Gaussian distribution, to which the classical correlator is matched, and the Laplace distribution, the true distribution in the problem, results in a minor degradation in performance. This phenomenon was also observed in \cite{Neri22}, in Example 4 therein. This concludes Example 1.
\subsection{Joint Optimization of the Signal and the Correlator}\label{Section3.2}
After characterizing the optimal correlator $\{w_t\}$ for a given signal, we proceed to a simultaneous optimization of the correlator and the signal, subject to the power constraint presented in \eqref{PowerConstraint} and the FA error constraint, for a fixed $\lambda$ to be optimized at a later stage.
\begin{flalign}\label{Joint_E_MD}
E_{\mbox{\tiny{MD}}}(E_{\mbox{\tiny{0}}}, \lambda, P_{\mbox{\tiny{s}}})+\lambda\theta &= \sup_{\left\{f_{\mbox{\tiny{S}}}:\;\mathbb{E}[S^2]\leq P_{\mbox{\tiny{s}}}\right\}}\sup_{\{f_{W|S}: E_{\mbox{\tiny{FA}}} \geq {E_{\mbox{\tiny{0}}}}\}}\left\{\mathbb{E}\left [\lambda W\cdot{S}-C_{\mbox{\tiny{V}}}(\lambda{W})\right ]\right\}
& \nonumber \\
&= \sup_{\left \{f_{\mbox{\tiny{W}}}:\; E_{\mbox{\tiny{FA}}} \geq {E_{\mbox{\tiny{0}}}}\right \}}\sup_{\left\{f_{S|W}:\;\mathbb{E}[S^2]\leq P_{\mbox{\tiny{s}}}\right\}}\left\{\mathbb{E}\left [\lambda W\cdot{S}-C_{\mbox{\tiny{V}}}(\lambda{W})\right ]\right\} 
& \nonumber \\
&\stackrel{(a)}{=} \sup_{\left \{f_{\mbox{\tiny{W}}}:\; E_{\mbox{\tiny{FA}}} \geq {E_{\mbox{\tiny{0}}}}\right \}}\left\{\mathbb{E}\left[\lambda W\cdot\sqrt{\frac{P_{\mbox{\tiny{s}}}}{\mathbb{E}[W^2]}}\cdot W - C_{\mbox{\tiny{V}}}(\lambda{W})\right ]\right\}  
& \nonumber \\
&=\sup_{\left \{f_{\mbox{\tiny{W}}}:\; E_{\mbox{\tiny{FA}}} \geq {E_{\mbox{\tiny{0}}}}\right \}}\left\{\lambda\cdot\sqrt{P_{\mbox{\tiny{s}}}\mathbb{E}[W^2]}-\mathbb{E}\left[C_{\mbox{\tiny{V}}}(\lambda{W})\right]\right\},
\end{flalign}
where (a) is due to the simple fact that, for a given $W$ and $P_{\mbox{\tiny{s}}}$, the 
correlation, $\mathbb{E}\left[W\cdot S\right]$ is maximized by $S=\sqrt{P_{\mbox{\tiny{s}}}/\mathbb{E}\left [W^2\right ]}\cdot W$.

Earlier we saw that the optimal correlator coefficients $W$ for a given signal $S$ are related by
\begin{equation}\label{32}
S = g\left (W; \alpha, \varphi, \lambda\right )
\end{equation}
and \eqref{32} must be satisfied for any signal $f_{\mbox{\tiny{S}}}$ and in particular, for one with an empirical PDF $f^*_{\mbox{\tiny{S}}}$ that maximizes the joint MD error exponent presented in \eqref{Joint_E_MD}. On the other hand, we have just seen that $S$ is a linear function of $W$ when it is optimised jointly with $W$. These two results settle together if and only if $W$ takes values only within a set of solutions, $\pazocal{S}(\xi)$, to the equation
\begin{align}\label{51}
\dot{C}_{\mbox{\tiny{V}}}(\lambda{w}) + \frac{\varphi\alpha}{\lambda}\dot{C}_{\mbox{\tiny{N}}}(\alpha{w}) = \xi w,
\end{align}
for some $\xi>0$. We note that $\pazocal{S}(\xi)$ is never an empty set as $w=0$ is always a solution to \eqref{51}.
Assuming that $g$ does not contain any linear segments, $\pazocal{S}(\xi)$ is a discrete set of intersection points between a non linear, monotonically ascending curve $g$ and a linear function $\xi w$. This implies that $f_{\mbox{\tiny{W}}}$ is a discrete PDF over the set of  $K$ values in $\pazocal{S}(\xi)$ ($K\triangleq |\pazocal{S}(\xi)|$), which we assume to be finite. Moreover, it is easy to see that $K$ is odd: For every non-zero $w'\in \pazocal{S}(\xi)$, $-w'$ is also a solution to \eqref{51}, since both $g$ and $\xi w$ are anti-symmetric around the origin. For $K=1$, the only solution is $w=0$ and therefore the test statistic is $\pazocal{T}=0$, independently from the hypothesis and the observations, so we can disregard this case and assume $K\ge 3$.  Now $w^{\mbox{\tiny{(i)}}}\in \pazocal{S}(\xi)$ can be indexed by $-\frac{K-1}{2}\le i \le\frac{K-1}{2}$, where $i$ is an integer and $w^{\mbox{\tiny{(i)}}}=-w^{\mbox{\tiny{(-i)}}}$. For the sake of simplicity of notation, this set will sometimes be denoted as $\{w^{(i)}\}$. The results in \cite{Neri22}, where $N$ was confined to be a Gaussian noise process, are a special case of the results here.

Now, $f_{\mbox{\tiny{W}}}$ can be written down as
\begin{equation}
f_{\mbox{\tiny{W}}}(w;\xi)=\sum_{w^{\mbox{\tiny{(i)}}}\in\pazocal{S}(\xi)}p_{\mbox{\tiny{i}}}\cdot \delta\left (w-w^{\mbox{\tiny{(i)}}}\right ), 
\end{equation}
where $p_i\triangleq \Pr\left \{W=w_{\mbox{\tiny{(i)}}}\right \}$, $-\frac{K-1}{2}\le i \le\frac{K-1}{2}$ and $\delta(\cdot )$ is Dirac's delta function.

We next show is that the vector $\textbf{p}\triangleq\left (p_{\mbox{\tiny{-(K-1)/2}}},p_{\mbox{\tiny{-(K-1)/2}+1}}, \dots, p_{\mbox{\tiny{(K-1)/2}}}\right )$ has a very simple structure. In fact, it has at most three non-zero entries. This property stems from the equivalence of our problem to a linear-programming (LP) problem. Consider the next chains of equalities, starting from the last line of \eqref{Joint_E_MD},
\begin{flalign}\label{34a}
E_{\mbox{\tiny{MD}}}(E_{\mbox{\tiny{0}}}, \lambda, P_{\mbox{\tiny{s}}})+\lambda\theta &=\sup_{\left \{f_{\mbox{\tiny{W}}}:\; E_{\mbox{\tiny{FA}}} \geq {E_{\mbox{\tiny{0}}}}\right \}}\left\{\lambda\cdot\sqrt{P_{\mbox{\tiny{s}}}\mathbb{E}[W^2]}-\mathbb{E}\left[C_{\mbox{\tiny{V}}}(\lambda{W})\right]\right\}
& \nonumber \\
&\stackrel{(a)}{=}\sup_{\left \{f_{\mbox{\tiny{W}}}:\; \sup_{\alpha\ge 0}\left \{E_{\mbox{\tiny{FA}}}(\alpha)\right \} \geq {E_{\mbox{\tiny{0}}}}\right \}}\left\{\lambda\cdot\sqrt{P_{\mbox{\tiny{s}}}\mathbb{E}[W^2]}-\mathbb{E}\left[C_{\mbox{\tiny{V}}}(\lambda{W})\right]\right\}
& \nonumber \\
&=\sup_{\left \{f_{\mbox{\tiny{W}}}:\;\cup_{\{\alpha\geq{0}\}}E_{\mbox{\tiny{FA}}}(\alpha) \geq {E_{\mbox{\tiny{0}}}}\right \}}\left\{\lambda\cdot\sqrt{P_{\mbox{\tiny{s}}}\mathbb{E}[W^2]}-\mathbb{E}\left[C_{\mbox{\tiny{V}}}(\lambda{W})\right]\right\}
& \nonumber \\
&=\sup_{\alpha\ge 0}\sup_{\left \{f_{\mbox{\tiny{W}}}:\;E_{\mbox{\tiny{FA}}}(\alpha) \geq {E_{\mbox{\tiny{0}}}}\right \}}\left\{\lambda\cdot\sqrt{P_{\mbox{\tiny{s}}}\mathbb{E}[W^2]}-\mathbb{E}\left[C_{\mbox{\tiny{V}}}(\lambda{W})\right]\right\}
& \nonumber \\
&=\sup_{\alpha\ge 0}\sup_{\xi> 0}\sup_{\left \{f_{\mbox{\tiny{W}}}(\cdot \;;\; \xi):\;E_{\mbox{\tiny{FA}}}(\alpha) \geq {E_{\mbox{\tiny{0}}}}\right \}}\left\{\lambda\cdot\sqrt{P_{\mbox{\tiny{s}}}\mathbb{E}[W^2]}-\mathbb{E}\left[C_{\mbox{\tiny{V}}}(\lambda{W})\right]\right\}
& \nonumber \\
&=\sup_{\alpha, \xi\ge 0}\sup_{\sigma^2_{\mbox{\tiny{W}}}\ge 0}\sup_{P\ge E_{\mbox{\tiny{0}}}}\left \{\lambda\cdot\sqrt{P_{\mbox{\tiny{s}}}\sigma^2_{\mbox{\tiny{W}}}}-\inf_{\left \{f_{\mbox{\tiny{W}}}(\cdot \;;\; \xi):\;E_{\mbox{\tiny{FA}}}(\alpha) = P,\; \mathbb{E}\left [W^2\right ]=\sigma^2_{\mbox{\tiny{W}}}\right \}}\left\{\mathbb{E}\left[C_{\mbox{\tiny{V}}}(\lambda{W})\right]\right\}\right \},
\end{flalign}
where (a) $E_{\mbox{\tiny{FA}}}(\alpha)\triangleq \alpha\theta-\mathbb{E}\left[C_{\mbox{\tiny{N}}}(\alpha W)\right]$.

The inner-most minimization problem in the last line of \eqref{34a} is indeed a LP problem expressed in standard form, with $m=3$ constraints and $K\ge 3$ variables to optimize:

\begin{equation}\label{LP1}
\begin{aligned}
\min_{\textbf{p}} \quad & \textbf{c}^T\textbf{p}\\
\textrm{s.t.} \quad & \textbf{Bp}= \textbf{d}\\
  &\textbf{p}\ge 0    \\
\end{aligned}
\end{equation}

where $\textbf{c}\triangleq\left (c_{\mbox{\tiny{-(k-1)/2}}}, c_{\mbox{\tiny{-(k-1)/2+1}}}, \dots, c_{\mbox{\tiny{(k-1)/2}}}\right )^T$, $c_i\triangleq C_{\mbox{\tiny{V}}}(\lambda w^{\mbox{\tiny{(i)}}})$, $\textbf{d} \triangleq \left (\alpha\theta-P, \sigma^2_{\mbox{\tiny{W}}}, 1\right )^T$ and
\begin{equation*}
\textbf{B}\triangleq
\begin{bmatrix}
\textbf{B}_{\mbox{\tiny{-(k-1)/2}}} & \textbf{B}_{\mbox{\tiny{-(k-1)/2+1}}} & \cdots &\textbf{B}_{\mbox{\tiny{(k-1)/2}}}
\end{bmatrix},
\end{equation*}
where $\textbf{B}_k\triangleq \left (C_{\mbox{\tiny{N}}}\left (\alpha w^{\mbox{\tiny{(k)}}}\right ), \left (w^{\mbox{\tiny{(k)}}}\right )^2, 1\right )^T$.

We can now apply a known result regarding LP problems: Under the conditions described above, the optimal solution $\textbf{p}^*$ has at most 3 non-zero entries (\textit{a basic solution}) in case the constraints are linearly independent, or less (\textit{a degenerate basic solution}). The interested reader is referred to \cite{LP} for a detailed exposition of the theory of LP.

Also, the search for $\textbf{p}^*$ can confined to vectors which are \textit{symmetric} in the sense that $p^*_{\mbox{\tiny{i}}}=p^*_{\mbox{\tiny{-i}}}$ for $-\frac{K-1}{2}\le i \le\frac{K-1}{2}$. To see why this is true, consider the following consideration: Due to the symmetry of the objective function $\mathbb{E}\left[C_{\mbox{\tiny{V}}}\left (\lambda W\right )\right]$ around the origin, for any vector $\textbf{p}$ we can construct a symmetric vector $\textbf{p}'$, whose $i$'th entry is given by $\frac{p_{\mbox{\tiny{i}}}+p_{\mbox{\tiny{-i}}}}{2}$. The value of the objective function, evaluated at $\textbf{p}$ and $\textbf{p}'$ is identical, thus only symmetric vectors can be considered to solve the linear program in \eqref{LP1}. This property will help us reduce the number of optimized parameters.

Once all the properties of $f^*_{\mbox{\tiny{W}}}$ are put together, it turns out we are looking for a distribution $f_{\mbox{\tiny{W}}}$ of the form
\begin{equation}
f_{\mbox{\tiny{W}}}(v;\xi)= p\cdot \delta\left (v+w\right )+(1-2p)\cdot \delta (v)+p\cdot\delta \left(v-w\right ),
\end{equation}
where $p\triangleq \frac{1-p_0}{2}$, $w\in\pazocal{S}(\xi)$ and $w>0$. It is also evident that such distribution of $W$ and $S$ is consistent with our initial presumption regarding the uniform convergence in \eqref{assumption1} and \eqref{assumption2}, since $W$ and $S$ each have at most three different values which appear at frequencies that are independent of $n$, similarly to the case in \eqref{27} in Example 1. Now, instead of trying to find the terms of $\pazocal{S}(\xi)$, which may be very difficult to carry out in practice, since it involves the evaluation of $g(\cdot , \alpha, \varphi,\lambda)$, where $\varphi$ itself may be a function of $\alpha$ and $\lambda$, we can optimize over $p$ and $w$ directly. A numerically manageable problem, equivalent to \eqref{Joint_E_MD}, is 
\begin{equation}
E_{\mbox{\tiny{MD}}}\left (E_{\mbox{\tiny{0}}}, \lambda, P_{\mbox{\tiny{s}}}\right )=\sup_{\left \{w,p:\; E_{\mbox{\tiny{FA}}}(w,p) \geq {E_{\mbox{\tiny{0}}}},\; w>0,\; 0\le p\le \frac{1}{2}\right \}}\left\{\lambda\cdot\sqrt{2P_{\mbox{\tiny{s}}}p}\cdot w-2p\cdot C_{\mbox{\tiny{V}}}(\lambda{w})\right\}-\lambda\theta,
\end{equation}
and 
\begin{equation}
E_{\mbox{\tiny{FA}}}(w,p)=\sup_{\alpha\ge 0 }\left \{\alpha\theta - 2p \cdot C_{\mbox{\tiny{N}}}(\alpha{w})\right \}.
\end{equation}

\noindent\textbf{Example 2.} 
Returning to the setting presented in Example 1, we shall present the characteristics of a detector that was optimized alongside the transmitted signal and compare it to the correlator that was optimized for the 4-ASK signalling, derived in Example 1. The problem of interest is
\begin{equation}\label{38}
E_{\mbox{\tiny{MD}}}\left (E_{\mbox{\tiny{0}}}, P_{\mbox{\tiny{s}}}\right ) = \sup_{\lambda\ge 0} E_{\mbox{\tiny{MD}}}\left (E_{\mbox{\tiny{0}}}, \lambda, P_{\mbox{\tiny{s}}}\right ),
\end{equation}
where we have set $P_{\mbox{\tiny{s}}}=5a^2$, so the comparison between the graphs obtained in \eqref{29} and \eqref{38} is proper. The results are presented in Figure \ref{fig:Example4ROC}.

\begin{figure}[H]
\centering
 \includegraphics[scale=0.5]{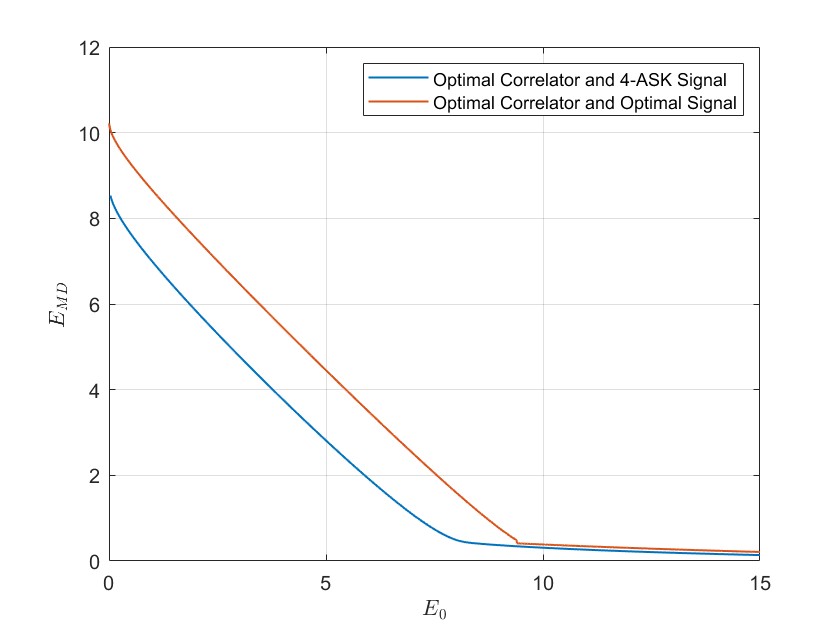}
\caption{MD error exponent as a function of $E_{\mbox{\tiny{0}}}$ for optimal correlators. The lower (blue) curve is when the signal is 4-ASK signal and the upper (red) curve is for the optimal signal for detection with correlation. Parameter values are identical to Example 1.}
\label{fig:Example4ROC}
\end{figure}
In Figure \ref{fig:Example4ROC}, we can see that the joint optimization of the signal and the correlator indeed yielded superior performance than the sub-optimal signal (4-ASK in this example) with an optimal correlator matched to it. The phenomenon explained in Example 1, that for large values of $E_{\mbox{\tiny{0}}}$, the optimal correlators perform very similarly to the classic correlator in the sense of FA to MD error exponent trade-off in the current setting of the problem, was repeated here as well. This concludes Example 2.
\section{Detectors Based on Linear Combinations of Correlation and Energy}\label{Section4}
We now discuss the more general case, where the test statistic to be compared to a threshold is given by
\begin{equation}\label{42}
\pazocal{T}=\sum_{t=1}^nw_tY_t+\gamma\sum_{t=1}^nY_t^2,
\end{equation}
and where $\gamma$ is optimized as opposed to setting $\gamma=0$, as was done in Section \ref{Section3}. We start our analysis by deriving an upper bound for the FA error exponent, similarly as in Section \ref{Section3}, using the Chernoff bound:
\begin{flalign}
P_{\mbox{\tiny{FA}}} &= \Pr\left\{ \sum_{t=1}^{n}w_{t}N_{t}+\gamma N_{t}^2\geq\theta{n} \right\}
& \nonumber \\
& \leq\inf_{\alpha\geq{0}}\left\{ e^{-\alpha\theta{n}}\cdot \prod_{t=1}^n e^{\ln\mathbb{E}\left [\exp\left (\alpha(w_t N_t+\gamma N_t^2\right )\right ]} \right\}
& \nonumber \\
&=\inf_{\alpha\geq{0}}\left\{ \exp\left [-\alpha\theta n+\sum_{t=1}^n \tilde{C}_{\mbox{\tiny{N}}}\left (\alpha w_t, \alpha\gamma\right )\right ] \right\}
& \nonumber \\
&=\exp\left\{ -\sup_{\alpha\geq{0}}\left [\alpha\theta n-\sum_{t=1}^n \tilde{C}_{\mbox{\tiny{N}}}\left (\alpha w_t, \alpha\gamma\right )\right ] \right\}.
\end{flalign}
The FA error exponent is therefore given by
\begin{equation}\label{44}
E_{\mbox{\tiny{FA}}}=\lim_{n\to\infty}\left \{\sup_{\alpha\geq{0}}\left [\alpha\theta-\frac{1}{n}\sum_{t=1}^n \tilde{C}_{\mbox{\tiny{N}}}\left (\alpha w_t, \alpha\gamma\right )\right ]\right \},
\end{equation}
provided that the limit exists. Under the same assumption as in Section \ref{Section3}, regarding the asymptotic behavior of $\textbf{w}$ and $\textbf{s}$, the limit in \eqref{44} exists, and so,
\begin{equation}\label{45}
E_{\mbox{\tiny{FA}}}(\alpha)\triangleq\lim_{n\to\infty}\left \{\left [\alpha\theta-\frac{1}{n}\sum_{t=1}^n \tilde{C}_{\mbox{\tiny{N}}}\left (\alpha w_t, \alpha\gamma\right )\right ]\right \}=\alpha\theta-\mathbb{E}\left [ \tilde{C}_{\mbox{\tiny{N}}}\left (\alpha W, \alpha\gamma\right )\right ],
\end{equation}
allowing to represent \eqref{44} more compactly as
\begin{equation}\label{46}
E_{\mbox{\tiny{FA}}}=\sup_{\alpha\geq{0}}\left \{\alpha\theta-\mathbb{E}\left [\tilde{C}_{\mbox{\tiny{N}}}\left (\alpha W, \alpha\gamma\right )\right ]\right \}.
\end{equation}
We find the upper bound for the MD exponent in a similar way to Section \ref{Section3},
\begin{flalign}\label{47}
P_{\mbox{\tiny{MD}}}&\stackrel{(a)}{=}\Pr\left\{\sum_{t=1}^n w_t(s_t+V_t)+\gamma\sum_{t=1}^n (s_t+V_t)^2<\theta n\right\}
& \nonumber \\
& \le \inf_{\lambda\ge 0}\left\{e^{\lambda n \theta}\cdot \mathbb{E}\left[\exp\left(-\lambda\left (\sum_{t=1}^n w_t(s_t+V_t)+\gamma\sum_{t=1}^n (s_t+V_t)^2\right)\right)\right]\right\}
& \nonumber \\
&= \inf_{\lambda\ge 0}\left\{e^{\lambda n \theta}\cdot \mathbb{E}\left[\exp\left(-\lambda\sum_{t=1}^n w_ts_t\right)\cdot \exp\left (-\lambda\sum_{t=1}^n w_t V_t\right)\cdot \exp\left (-\lambda\gamma\sum_{t=1}^n (s_t+V_t)^2\right)\right]\right\}
& \nonumber \\
&= \inf_{\lambda\ge 0}\Bigg\{e^{\lambda\left(n \theta-\sum_{t=1}^n w_ts_t\right)}\cdot\mathbb{E}\Bigg[\exp\left (-\lambda\sum_{t=1}^n w_t V_t\right)\cdot \exp\left(-\lambda\gamma\sum_{t=1}^n s_t^2\right )
& \nonumber \\
&\quad\quad\quad\cdot \exp\left (-2\lambda\gamma\sum_{t=1}^n s_tV_t\right )\cdot \exp\left(-\lambda\gamma\sum_{t=1}^n V_t^2\right )\Bigg]\Bigg\}
& \nonumber \\
&= \inf_{\lambda\ge 0}\left\{e^{\lambda\left(n \theta-\sum_{t=1}^n \left (w_ts_t+\gamma s_t^2\right )\right)}\cdot\prod_{t=1}^n \mathbb{E}\left[e^{-(\lambda w_t+2\lambda\gamma s_t) V_t-\lambda\gamma V_t^2}\right]\right\}
& \nonumber \\
&= \inf_{\lambda\ge 0}\left\{\exp\left [\lambda\left(n \theta-\sum_{t=1}^n \left (w_ts_t+\gamma s_t^2\right )\right)+\sum_{t=1}^n \tilde{C}_{\mbox{\tiny{V}}}\left (-\left (\lambda w_t+2\lambda\gamma s_t\right ),-\lambda\gamma\right )\right]\right\}
& \nonumber \\
&\stackrel{(c)}{=} \exp\left\{\inf_{\lambda\ge 0}\left [\lambda\left(n \theta-\sum_{t=1}^n \left (w_ts_t+\gamma s_t^2\right )\right)+\sum_{t=1}^n \tilde{C}_{\mbox{\tiny{V}}}\left (\lambda w_t+2\lambda\gamma s_t,-\lambda\gamma\right )\right]\right\}
& \nonumber \\
&= \exp\left\{-n\cdot \sup_{\lambda\ge 0}\left [-\lambda\theta+\frac{1}{n}\sum_{t=1}^n \left(\lambda w_ts_t+\lambda\gamma s_t^2+\tilde{C}_{\mbox{\tiny{V}}}\left (\lambda w_t+2\lambda\gamma s_t,-\lambda\gamma\right )\right )\right]\right\},
\end{flalign}
where in (a) we have defined $V_t\triangleq N_t+Z_t$, (b) is because $\{V_t\}$ is an i.i.d. noise process and (c) follows from the symmetry of $f_{\mbox{\tiny{V}}}$ around the origin.

The MD exponent resulting from \eqref{47} is
\begin{equation}\label{48}
{E}_{\mbox{\tiny{MD}}}\triangleq\lim_{n\to\infty}\left \{\sup_{\lambda\ge 0}\left [-\lambda\theta+\frac{1}{n}\sum_{t=1}^n \left(\lambda w_ts_t+\lambda\gamma s_t^2+\tilde{C}_{\mbox{\tiny{V}}}\left (\lambda w_t+2\lambda\gamma s_t,-\lambda\gamma\right )\right )\right]\right \},
\end{equation}
provided that the limit exists. The previous assumption regarding the asymptomatic behavior of $\textbf{w}$ and $\textbf{s}$ ensures that the limit in \eqref{49} exists and that the convergence therein is uniform in $\lambda\ge 0$ and $\gamma\ge 0$.
\begin{flalign}\label{49}
{E}_{\mbox{\tiny{MD}}}(\lambda)&\triangleq\lim_{n\to\infty}\left \{\left [-\lambda\theta+\frac{1}{n}\sum_{t=1}^n \left(\lambda w_ts_t+\lambda\gamma s_t^2+\tilde{C}_{\mbox{\tiny{V}}}\left (\lambda w_t+2\lambda\gamma s_t,-\lambda\gamma\right )\right )\right]\right \}
& \nonumber \\
&=\mathbb{E}\left [\lambda\left (WS+\gamma S^2-\theta\right )+\tilde{C}_{\mbox{\tiny{V}}}\left (\lambda W+2\lambda\gamma S,-\lambda\gamma\right )\right ],
\end{flalign}
and therefore \eqref{48} is equivalent to
\begin{equation}\label{50}
{E}_{\mbox{\tiny{MD}}}=\sup_{\lambda\ge 0} \left \{\mathbb{E}\left [\lambda\left (WS+\gamma S^2-\theta\right )+\tilde{C}_{\mbox{\tiny{V}}}\left (\lambda W+2\lambda\gamma S,-\lambda\gamma\right )\right ]\right \}.
\end{equation}
It is worthwhile to mention that the class of distributions of $V$ with a finite $\tilde{C}_{\mbox{\tiny{V}}}$ in a neighborhood around the origin is smaller than the one with a finite $C_{\mbox{\tiny{V}}}$ in an interval around the origin. For example, in the setting presented in Examples 1 and 2, where it was assumed that $N$ had a Laplace distribution, \eqref{46} and \eqref{50} are diverging and the Chernoff bounding technique is not applicable any more. In such scenarios, techniques of heavy-tailed large deviations are needed, see for example \cite{LargeDeviation}.
\subsection{Optimal Correlation and Energy Detector for a Given Signal}
We now extend the results of Section \ref{Section3} regarding the asymptotically optimal correlator weights for a given signal for the detector presented in \eqref{42}, for a fixed $\lambda$ and $\gamma$ that will be optimised separately at a later stage. Again, we seek the optimal conditional PDF, $f_{W|S}$, such that, 
\begin{equation}
{E}_{\mbox{\tiny{MD}}}(\lambda)=\int_{-\infty}^\infty f_{\mbox{\tiny{S}}}(s)\cdot \mathbb{E}\left [\lambda\left (Ws+\gamma s^2\right )+\tilde{C}_{\mbox{\tiny{V}}}\left (\lambda W+2\lambda\gamma s,-\lambda\gamma\right )|S=s\right ]\mbox{d}s-\lambda\theta,
\end{equation}
is maximal w.r.t. $f_{W|S}$ for fixed $\lambda$ and $\gamma$, subject to the FA error constraint,
\begin{equation}
{E}_{\mbox{\tiny{MD}}}\ge E_{\mbox{\tiny{0}}},
\end{equation}
where $E_{\mbox{\tiny{0}}}$ was defined in \eqref{19}.
\begin{theorem}\label{Theorem2}
Let the assumptions of Section \ref{Section2}, \ref{Section4} hold and let $N$ be a non degenerate RV. Define the function $\tilde{g}$ as
\begin{equation}
\tilde{g}\left (w;\alpha,\varphi,\lambda, \gamma\right )\triangleq \tilde{C}_{\mbox{\tiny{Vx}}}\left (\lambda w+2\lambda\gamma s ,-\lambda\gamma\right )+\frac{\alpha\varphi}{\lambda}\tilde{C}_{\mbox{\tiny{Nx}}}\left (\alpha w, \alpha\gamma\right ),
\end{equation}
where $\tilde{C}_{\mbox{\tiny{Vx}}}$ and $\tilde{C}_{\mbox{\tiny{Nx}}}$ denote the partial derivatives of $\tilde{C}_{\mbox{\tiny{V}}}(x,y)$ and $\tilde{C}_{\mbox{\tiny{N}}}(x,y)$ w.r.t $x$, respectively. Further assume, that there exists $\varphi\ge 0$ (possibly, dependent on $\lambda$, $\alpha$ and $\gamma$) that satisfies:
\begin{equation}
\mathbb{E}\left [\tilde{C}_{\mbox{\tiny{N}}}\left (\alpha \cdot \tilde{g}^{\mbox{\tiny{-1}}}\left (s; \alpha, \varphi, \lambda, \gamma\right ), \alpha\gamma\right )\right ]=\alpha\theta-E_{\mbox{\tiny{0}}},
\end{equation}
where $\tilde{g}^{\mbox{\tiny{-1}}}$ is the inverse function of $\tilde{g}$ (as a function $w$). If no such $\varphi$ exists, set $\varphi=0$. The optimal conditional PDF is given by
\begin{equation}
f^*_{W|S}(w|s)=\delta\left (w-\tilde{g}^{\mbox{\tiny{-1}}}\left (s; \alpha, \varphi, \lambda, \gamma\right )\right ),
\end{equation}
where $\delta(\cdot )$ is the Dirac delta function and 
\begin{equation}\label{56}
 \alpha=\arg\max_{\alpha\ge 0}\left \{\mathbb{E}\left  [\lambda\tilde{g}^{\mbox{\tiny{-1}}}\left (S;\alpha,\varphi,\lambda, \gamma\right)\cdot S+\lambda\gamma S^2
-\tilde{C}_{\mbox{\tiny{V}}}\left (\lambda \tilde{g}^{\mbox{\tiny{-1}}}\left (S;\alpha,\varphi,\lambda, \gamma\right)+2\lambda\gamma S ,-\lambda\gamma\right )\right  ]\right \}
\end{equation} 
\end{theorem}
\begin{proof}
We first show that $g$ is invertible. Due to Lemma \ref{Lemma3} in the Appendix, both $\tilde{C}_{\mbox{\tiny{N}}}$ and $\tilde{C}_{\mbox{\tiny{V}}}$ are strictly convex (see proof therein), so their second partial derivatives w.r.t. $x$ are positive. Thus, $C_{\mbox{\tiny{Nx}}}$ and $C_{\mbox{\tiny{Vx}}}$ are monotonically increasing and so $\tilde{g}$ has an inverse function. Next, consider the following derivation
\begin{flalign}\label{57a}
{E}_{\mbox{\tiny{MD}}}\left (E_{\mbox{\tiny{FA}}},\lambda,\gamma\right )+\lambda\theta &=\sup_{f_{W|S}:\; {E}_{\mbox{\tiny{FA}}} \geq E_{\mbox{\tiny{0}}}}\left\{\int_{-\infty}^\infty f_{\mbox{\tiny{S}}}(s)\cdot\mathbb{E}\left [\lambda\left (Ws+\gamma s^2\right )-\tilde{C}_{\mbox{\tiny{V}}}\left (\lambda W+2\lambda\gamma s ,-\lambda\gamma\right )\right ]\mbox{d}s\right \}
& \nonumber \\
&= \sup_{f_{W|S}}\inf_{\rho\ge 0}\bigg\{\int_{-\infty}^\infty f_{\mbox{\tiny{S}}}(s)\cdot\mathbb{E}\left [\lambda\left (Ws+\gamma s^2\right )-\tilde{C}_{\mbox{\tiny{V}}}\left (\lambda W+2\lambda\gamma s ,-\lambda\gamma\right )\right ]\mbox{d}s
& \nonumber \\
&~~~~~~~~~~~~~~~~+\rho\left [{E}_{\mbox{\tiny{FA}}}- E_{\mbox{\tiny{0}}}\right ]\bigg \}
& \nonumber \\
&= \sup_{f_{W|S}}\inf_{\rho\ge 0}\sup_{\alpha\ge 0}\bigg\{\int_{-\infty}^\infty f_{\mbox{\tiny{S}}}(s)\cdot\mathbb{E}\left [\lambda\left (Ws+\gamma s^2\right )-\tilde{C}_{\mbox{\tiny{V}}}\left (\lambda W+2\lambda\gamma s ,-\lambda\gamma\right )\right ]\mbox{d}s
& \nonumber \\
&~~~~~~~~~~~~~~~~+\rho\left [\alpha\theta-\mathbb{E}\left [ \tilde{C}_{\mbox{\tiny{N}}}\left (\alpha W, \alpha\gamma\right )\right ]-E_{\mbox{\tiny{0}}}\right ]\bigg \}
& \nonumber \\
&\stackrel{(a)}{=} \sup_{f_{W|S}}\sup_{\alpha\ge 0}\inf_{\rho\ge 0}\bigg\{\int_{-\infty}^\infty f_{\mbox{\tiny{S}}}(s)\cdot\mathbb{E}\left [\lambda\left (Ws+\gamma s^2\right )-\tilde{C}_{\mbox{\tiny{V}}}\left (\lambda W+2\lambda\gamma s ,-\lambda\gamma\right )\right ]\mbox{d}s
& \nonumber \\
&~~~~~~~~~~~~~~~~+\rho\left [\alpha\theta-\mathbb{E}\left [ \tilde{C}_{\mbox{\tiny{N}}}\left (\alpha W, \alpha\gamma\right )\right ]-E_{\mbox{\tiny{0}}}\right ]\bigg \}
& \nonumber \\
&= \sup_{\alpha\ge 0}\sup_{f_{W|S}}\inf_{\rho\ge 0}\bigg\{\int_{-\infty}^\infty f_{\mbox{\tiny{S}}}(s)\cdot\mathbb{E}\left [\lambda\left (Ws+\gamma s^2\right )-\tilde{C}_{\mbox{\tiny{V}}}\left (\lambda W+2\lambda\gamma s ,-\lambda\gamma\right )\right ]\mbox{d}s
& \nonumber \\
&~~~~~~~~~~~~~~~~+\rho\left [\alpha\theta-\mathbb{E}\left [ \tilde{C}_{\mbox{\tiny{N}}}\left (\alpha W, \alpha\gamma\right )\right ]-E_{\mbox{\tiny{0}}}\right ]\bigg \}
& \nonumber \\
&\stackrel{(b)}{=} \sup_{\alpha\ge 0}\inf_{\rho\ge 0}\sup_{f_{W|S}}\bigg\{\int_{-\infty}^\infty f_{\mbox{\tiny{S}}}(s)\cdot\mathbb{E}\left [\lambda\left (Ws+\gamma s^2\right )-\tilde{C}_{\mbox{\tiny{V}}}\left (\lambda W+2\lambda\gamma s ,-\lambda\gamma\right )\right ]\mbox{d}s
& \nonumber \\
&~~~~~~~~~~~~~~~~+\rho\left [\alpha\theta-\mathbb{E}\left [ \tilde{C}_{\mbox{\tiny{N}}}\left (\alpha W, \alpha\gamma\right )\right ]-E_{\mbox{\tiny{0}}}\right ]\bigg \}
& \nonumber \\
&\stackrel{(c)}{=} \sup_{\alpha\ge 0}\inf_{\rho\ge 0}\bigg\{\int_{-\infty}^\infty f_{\mbox{\tiny{S}}}(s)\cdot\sup_w\Big [\lambda ws-\tilde{C}_{\mbox{\tiny{V}}}\left (\lambda w+2\lambda\gamma s ,-\lambda\gamma\right )-\tilde{C}_{\mbox{\tiny{N}}}\left (\alpha w, \alpha\gamma\right )
& \nonumber \\
&~~~~~~~~~~~~~~~~+\lambda\gamma s^2\Big ]\mbox{d}s+\rho\left [\alpha\theta - E_{\mbox{\tiny{0}}}\right ]\bigg \}
& \nonumber \\
&\stackrel{(d)}{=} \sup_{\alpha\ge 0}\inf_{\rho\ge 0}\bigg\{\int_{-\infty}^\infty f_{\mbox{\tiny{S}}}(s)\Big [\lambda\tilde{g}^{\mbox{\tiny{-1}}}\left (s;\alpha,\rho,\lambda, \gamma\right)s-\tilde{C}_{\mbox{\tiny{V}}}\left (\lambda \tilde{g}^{\mbox{\tiny{-1}}}\left (s;\alpha,\rho,\lambda, \gamma\right)+2\lambda\gamma s ,-\lambda\gamma\right )
& \nonumber \\
&~~~~~~~~~~~~~~~~+\lambda\gamma s^2-\rho\cdot \tilde{C}_{\mbox{\tiny{N}}}\left (\alpha \tilde{g}^{\mbox{\tiny{-1}}}\left (s;\alpha,\rho,\lambda, \gamma\right), \alpha\gamma\right )\Big ]\mbox{d}s+\rho\left [\alpha\theta - E_{\mbox{\tiny{0}}}\right ]\bigg \}
& \nonumber \\
&\stackrel{(e)}{\le} \sup_{\alpha\ge 0}\bigg\{\int_{-\infty}^\infty f_{\mbox{\tiny{S}}}(s)\Big [\lambda\tilde{g}^{\mbox{\tiny{-1}}}\left (s;\alpha,\varphi,\lambda, \gamma\right)\cdot s-\tilde{C}_{\mbox{\tiny{V}}}\left (\lambda \tilde{g}^{\mbox{\tiny{-1}}}\left (s;\alpha,\varphi,\lambda, \gamma\right)+2\lambda\gamma s ,-\lambda\gamma\right )
& \nonumber \\
&~~~~~~~~~~~~~~~~+\lambda\gamma s^2-\rho\cdot \tilde{C}_{\mbox{\tiny{N}}}\left (\alpha \tilde{g}^{\mbox{\tiny{-1}}}\left (s;\alpha,\varphi,\lambda, \gamma\right), \alpha\gamma\right )\Big ]\mbox{d}s+\rho\left [\alpha\theta - E_{\mbox{\tiny{0}}}\right ]\bigg \}
& \nonumber \\
&= \sup_{\alpha\ge 0}\bigg\{\int_{-\infty}^\infty f_{\mbox{\tiny{S}}}(s)\Big  [\lambda\tilde{g}^{\mbox{\tiny{-1}}}\left (s;\alpha,\varphi,\lambda, \gamma\right)\cdot s-\tilde{C}_{\mbox{\tiny{V}}}\left (\lambda \tilde{g}^{\mbox{\tiny{-1}}}\left (s;\alpha,\varphi,\lambda, \gamma\right)+2\lambda\gamma s ,-\lambda\gamma\right )
& \nonumber \\
&~~~~~~~~~~~~~~~~+\lambda\gamma s^2\Big  ]\mbox{d}s+\rho\left [\alpha\theta - \int_{-\infty}^\infty f_{\mbox{\tiny{S}}}(s)\tilde{C}_{\mbox{\tiny{N}}}\left (\alpha \tilde{g}^{\mbox{\tiny{-1}}}\left (s;\alpha,\varphi,\lambda, \gamma\right), \alpha\gamma\right )\mbox{d}s-E_{\mbox{\tiny{0}}}\right ]\bigg \}
& \nonumber \\
&\stackrel{(f)}{=} \sup_{\alpha\ge 0}\left \{\int_{-\infty}^\infty f_{\mbox{\tiny{S}}}(s)\left  [\lambda\left (\tilde{g}^{\mbox{\tiny{-1}}}\left (s;\alpha,\varphi,\lambda, \gamma\right)s+\gamma s^2\right )-\tilde{C}_{\mbox{\tiny{V}}}\left (\lambda \tilde{g}^{\mbox{\tiny{-1}}}\left (s;\alpha,\varphi,\lambda, \gamma\right)+2\lambda\gamma s ,-\lambda\gamma\right )\right  ]\mbox{d}s\right \}
& \nonumber \\
&= \sup_{\alpha\ge 0}\Big \{\mathbb{E}\Big  [\lambda\tilde{g}^{\mbox{\tiny{-1}}}\left (S;\alpha,\varphi,\lambda, \gamma\right)\cdot S+\lambda\gamma S^2
-\tilde{C}_{\mbox{\tiny{V}}}\left (\lambda \tilde{g}^{\mbox{\tiny{-1}}}\left (S;\alpha,\varphi,\lambda, \gamma\right)+2\lambda\gamma S ,-\lambda\gamma\right )\Big  ]\Big \}
& \nonumber \\
\end{flalign}
where (a) is because the objective function is linear in $\rho$ and convex in $\alpha$, (b) is since the objective function is linear $\rho$ and linear in $ f_{W|S}$, (c) the unconstrained maximum of the conditional expectation is attained when $f_{W|S}$ concentrates all its mass on the maximizing $w$, (d) is because the maximum of a concave function of $w$ is achieved at the point of zero-derivative, $w=\tilde{g}^{\mbox{\tiny{-1}}}\left (s;\alpha,\rho,\lambda, \gamma\right)$, (e) follows the first postulate that $\varphi\ge 0$ and (f) is by the presumption that we set $\varphi$ such that \eqref{21a} holds or we set it to zero.

We have therefore shown that the upper bound on the constrained maximum in the first line of \eqref{57a} is attained by $W=\tilde{g}^{\mbox{\tiny{-1}}}\left (S;\alpha, \varphi, \lambda,\gamma\right )$ with probability one under $f_{W|S}$, where $\alpha$ is defined in \eqref{56}.
\end{proof}
Clearly, Theorem \ref{Theorem1} is a special case of Theorem \ref{Theorem2} where $\gamma=0$.

\noindent\textbf{Example 3.} In this example, a similar setting of Example 1 is kept, except $N$ has a uniform distribution in the interval $[-B, B]$, where $B>0$ is a given constant, rather than a Laplace distribution and the transmitted signal is ternary symmetric, i.e, the $s_t=0$ for half the time and $s_t=\pm a$ along the other half. The CGF of $N$ is given by
\begin{equation}
C_{\mbox{\tiny{N}}}(v)=\ln\left (\frac{\sinh(Bv)}{Bv}\right ).
\end{equation}
Next, we derive the asymptotically optimal correlator for a given signal ($\gamma=0$) and compare it to the asymptotically optimal energy-correlation detector ($\gamma\ne 0$). For $\gamma=0$, $N$ and $V$ are substituted properly in \eqref{29} and $\beta$ is set to $0$. For $\gamma\ne0$, the MD error exponent is given by
\begin{equation}\label{57}
{E}_{\mbox{\tiny{MD}}}(E_{\mbox{\tiny{0}}})=\frac{1}{2}\cdot \sup_{\gamma,\lambda,\delta,\alpha\ge 0}\bigg\{a\lambda\delta-a^2\lambda\gamma-\tilde{C}_{\mbox{\tiny{V}}}(\lambda \delta,-\lambda\gamma)-\tilde{C}_{\mbox{\tiny{V}}}(0,-\lambda\gamma)-\frac{\lambda}{\alpha}[2E_{\mbox{\tiny{0}}}+\tilde{C}_{\mbox{\tiny{N}}}\left (\alpha\delta,\alpha\gamma\right )+\tilde{C}_{\mbox{\tiny{N}}}\left (0,\alpha\gamma\right )]\bigg \},
\end{equation}
where the  moment-generating function (MGF) of $N$ is given by
\begin{equation}
\tilde{M}_{\mbox{\tiny{N}}}(x,y)=\frac{1}{4B}\cdot \sqrt{\frac{\pi}{|y|}}\cdot \exp\left (-\frac{x^2}{4y}\right )\cdot \left\{\begin{matrix}
\textrm{erf}\left (\sqrt{|y|}B-\frac{x}{2\sqrt{|y|}}
\right )+\textrm{erf}\left (\sqrt{|y|}B+\frac{x}{2\sqrt{|y|}}\right ), ~~~~ y<0 \\
\textrm{erfi}\left (\sqrt{y}B+\frac{x}{2\sqrt{y}}\right )+\textrm{erfi}\left (\sqrt{y}B-\frac{x}{2\sqrt{y}}\right ), ~~~~ y>0
\end{matrix}\right.
\end{equation}
and the MGF of $V$ is given by 
\begin{equation}
\tilde{M}_{\mbox{\tiny{V}}}(x,y)=\frac{e^{yz_0^2}}{2}\left [\tilde{M}_{\mbox{\tiny{N}}}\left (x+2yz_0, y\right )\cdot e^{xz_0}+\tilde{M}_{\mbox{\tiny{N}}}\left (x-2yz_0, y\right )\cdot e^{-xz_0}\right ],
\end{equation}
where 
\begin{equation}
\textrm{erf}(x)\triangleq \int_0^x\frac{2}{\sqrt{\pi}}e^{-t^2}\mbox{d}t,\quad \textrm{erfi}(x)\triangleq \int_0^x\frac{2}{\sqrt{\pi}}e^{t^2}\mbox{d}t.
\end{equation}
The maximum allowed FA exponent constraint was substituted similarly to Example 1. Results of the optimization problems in \eqref{29} and \eqref{57} for $E_{\mbox{\tiny{0}}}\in[0,3]$ are presented in Figure \ref{fig:Example5ROC} for specific values of $a$, $z_{\mbox{\tiny{0}}}$ and $B$.
\begin{figure}[H]
\centering
\includegraphics[scale=0.5]{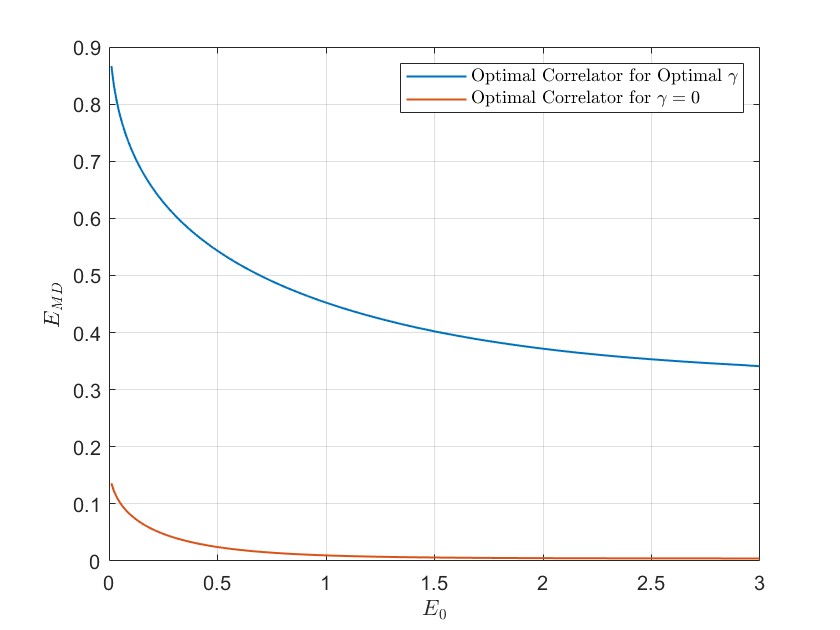}
\caption{MD error exponents as functions of $E_{\mbox{\tiny{0}}}$ for the optimal correlator and the optimal correlation-energy detector with parameter values $a=6$, $z_{\mbox{\tiny{0}}}=7$ and $B=5$.}
\label{fig:Example5ROC}
\end{figure}
As expected, the detector with the optimal weight of the energy term, $\gamma$, outperforms the one without any energy term ($\gamma=0$). We note that for both detectors, a very slow decrease of the MD error exponent is observed for large values of $E_{\mbox{\tiny{0}}}$. In this example, the support of the PDFs of the channel noise and the SIN do not overlap. Thus, for a small enough prescribed FA error probability , $\theta$ converges to a critical threshold $\theta_c$, above which no false alarm errors occur at all. In such a scenario, an infinitesimal increase in $\theta$, results in an infinitesimal decrease of the MD error exponent and a huge increase in the FA error exponent. To demonstrate this phenomenon, we plot $E_{\mbox{\tiny{0}}}$ as a function of $\theta$, and see if there exists $\theta_c$, such that if $\theta\rightarrow\theta_c$, then $E_{\mbox{\tiny{0}}}\rightarrow\infty$. As we approach this problem, we return to \eqref{29} and note the pair $(\delta,\beta,\theta)$ that the receiver designer wishes to find is invariant to scaling. Assume that one calculates the FA and MD error probability for a given $(\delta,\beta,\theta)$. Since all of these parameters are multiplied by $\lambda$ and by $\alpha$ in \eqref{29}, it is easy to see that if $(\delta,\beta,\theta)$ satisfies \eqref{29}, so does $(\delta\cdot x,\beta\cdot x,\theta\cdot x)$ for every $x>0$, as every scaling factor $x$ can be absorbed in the Chernoff parameters. So, for the desired plot to be meaningful, we must eliminate this degree of freedom. We do so by noticing that \eqref{63a} attains the same optimal value as \eqref{29}.
\begin{align}\label{63a}
E_{\mbox{\tiny{MD}}}\left (E_{\mbox{\tiny{0}}}, f_{\mbox{\tiny{S}}}\right )=\frac{1}{2}\cdot \sup_{\lambda,\; \beta,\; \alpha\ge 0}\left\{a\lambda+3a\lambda\beta - C_{\mbox{\tiny{V}}}(\lambda)- C_{\mbox{\tiny{V}}}(\lambda\beta)-\frac{\lambda}{\alpha}\left [2E_{\mbox{\tiny{0}}}+C_{\mbox{\tiny{N}}}(\alpha)+C_{\mbox{\tiny{N}}}(\alpha\beta)\right ]\right\},
\end{align}
where $\theta= \frac{1}{2\alpha}\cdot \left [2E_{\mbox{\tiny{0}}}+C_{\mbox{\tiny{N}}}(\alpha)+C_{\mbox{\tiny{N}}}(\alpha\beta)\right ]$ and in this example we substitute $\beta=0$. The results are shown in Figure \ref{fig:Example5threshold}.
\begin{figure}[H]
\centering
\includegraphics[scale=0.5]{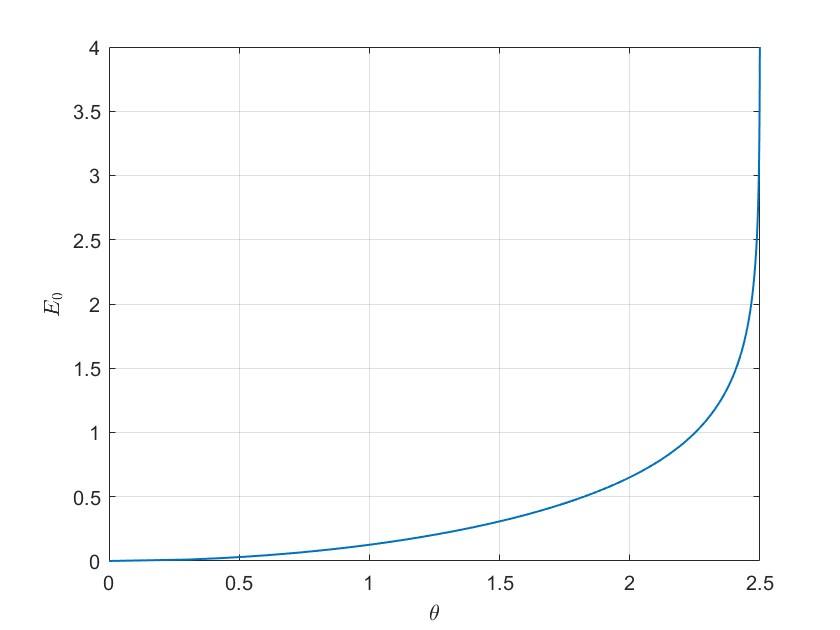}
\caption{The threshold $\theta$ as a function of $E_{\mbox{\tiny{0}}}$ for the classical optimal correlator for a given signal. (same parameter values as Figure \ref{fig:Example5ROC}).}
\label{fig:Example5threshold}
\end{figure}
We see that in this example, $\theta_c=2.5$ for the linear detector.
\subsection{Joint Optimization of the Signal and the Correlator}\label{Section4.2}
In the same spirit of Section \ref{Section3.2}, we now derive the MD error exponent when the signal and the correlator are optimized jointly, for fixed $\lambda$ and $\gamma$ which are later optimized, subject to the power and the FA error constraints.
\begin{flalign}\label{63}
{E}_{\mbox{\tiny{MD}}}(E_{\mbox{\tiny{0}}}, \lambda, \gamma)+\lambda\theta&=\sup_{\left \{f_W:\; {E}_{\mbox{\tiny{FA}}}\ge E_{\mbox{\tiny{0}}}\right \}}\sup_{\left \{f_{S|W}:\; \mathbb{E}[S^2]\le P_{\mbox{\tiny{s}}}\right \}}\left \{\mathbb{E}\left [\lambda W\cdot S+\lambda \gamma S^2-\tilde{C}_{\mbox{\tiny{V}}}(\lambda W+2\lambda\gamma S,-\lambda\gamma)\right ]\right \}
& \nonumber \\
&=\sup_{P\le P_{\mbox{\tiny{s}}}}\sup_{C\ge E_{\mbox{\tiny{0}}}}\sup_{\alpha\ge 0}\sup_{\left \{f_{\mbox{\tiny{WS}}}:\; {E}_{\mbox{\tiny{FA}}}(\alpha)= C,\;\mathbb{E}[S^2]= P\right \}}\left \{\mathbb{E}\left [\lambda W\cdot S+\lambda \gamma S^2-\tilde{C}_{\mbox{\tiny{V}}}(\lambda W+2\lambda\gamma S,-\lambda\gamma)\right ]\right \},
\end{flalign}
where ${E}_{\mbox{\tiny{FA}}}(\alpha)$ in the last line of \eqref{63} was defined in \eqref{45}.

Focussing on the inner-most optimization problem in \eqref{63}, we would like to treat it as a LP, which requires $f_{\mbox{\tiny{WS}}}$ to be discrete. To do so, we will first  consider a discrete approximation of $f_{\mbox{\tiny{WS}}}$, denoted $f'_{\mbox{\tiny{WS}}}$, and then we will show that in the limit of fine quantization, the approximation becomes accurate. Since we assume that the amplitudes of $W$ and $S$ are bounded by $D_{\mbox{\tiny{1}}}$ and $D_{\mbox{\tiny{2}}}$, respectively, adding quantization levels improves the approximation. The support of $f_{\mbox{\tiny{WS}}}(w,s)$ therefore is the rectangle $|w|\le D_{\mbox{\tiny{1}}}, |s|\le D_{\mbox{\tiny{2}}}$ and we propose $f'_{\mbox{\tiny{WS}}}$ of the form
\begin{equation}
f'_{\mbox{\tiny{WS}}}(i,j; k)\triangleq\Pr\left \{W=\frac{i}{k}\cdot D_{\mbox{\tiny{1}}},S=\frac{j}{k}\cdot D_{\mbox{\tiny{2}}}\right \}, \quad -k\le i,j \le k,
\end{equation}
where $k\ge 1$ is a give integer. We denote $w^{\mbox{\tiny{(i)}}}\triangleq\frac{i}{k}\cdot D_{\mbox{\tiny{1}}}$ and $s^{\mbox{\tiny{(j)}}}\triangleq\frac{j}{k}\cdot D_{\mbox{\tiny{2}}}$. The larger is $k$, the more accurate the approximate MD exponent becomes. The original optimization problem over $f_{\mbox{\tiny{WS}}}$ is essentially equivalent to a maximization problem over the $f'_{\mbox{\tiny{WS}}}$ when $k$ is infinitely large. The equivalent problem is in turn a LP with $M=3$ equality constraints:
\begin{enumerate}
\item
$\alpha\theta+\sum_{i=-k}^k \Pr\{W=w^{\mbox{\tiny{(i)}}}\}\cdot \tilde{C}_{\mbox{\tiny{N}}}(\alpha w^{\mbox{\tiny{(i)}}},\alpha\gamma)=C$.
\item
$\sum_{i=-k}^k \Pr\{S=s^{\mbox{\tiny{(i)}}}\}\cdot \left (s^{\mbox{\tiny{(i)}}}\right )^2=P$.
\item
$\sum_{i,j=-k}^k \Pr\{W=w^{\mbox{\tiny{(i)}}}, S=s^{\mbox{\tiny{(j)}}}\}=1$.
\end{enumerate} 
And $\Pr\{W=w^{\mbox{\tiny{(i)}}}, S=s^{\mbox{\tiny{(j)}}}\}\ge 0$ for all $i$, $j$.

We now introduce a square matrix $\textbf{P}\in \mathbb{R}^{(2k+1)\times (2k+1)}$, which its $(i,j)$'th entry is given by $\Pr\{W=w^{\mbox{\tiny{(i)-k-1}}}, S=s^{\mbox{\tiny{(j)-k-1}}}\}$ and the vector $\textbf{p}\in \mathbb{R}^{(2k+1)^2\times 1}$, which is constructed from the columns of $\textbf{P}$ stacked one below the other, starting from left to right, to a single column vector. A matrix $\textbf{A}\in\mathbb{R}^{3\times (2k+1)^2}$ is defined as well,
\begin{equation*}
\textbf{A}\triangleq
\begin{bmatrix}
\textbf{A}_{\mbox{\tiny{-k}}} & \textbf{A}_{\mbox{\tiny{-k+1}}} & \cdots &\textbf{A}_{\mbox{\tiny{k}}}
\end{bmatrix},
\end{equation*}
where $\textbf{A}_{\mbox{\tiny{k}}}\in\mathbb{R}^{3\times (2k+1)}$ is defined as follows,
\begin{equation*}
\textbf{A}_{\mbox{\tiny{k}}}\triangleq
\begin{bmatrix}
c^{\mbox{\tiny{(-k)}}} & c^{\mbox{\tiny{(-k+1)}}} & \cdots &c^{\mbox{\tiny{(k)}}} \\
\left (s^{\mbox{\tiny{(k)}}}\right )^2 & \left (s^{\mbox{\tiny{(k)}}}\right )^2 & \cdots & \left (s^{\mbox{\tiny{(k)}}}\right )^2 \\
1 & 1 & \cdots & 1
\end{bmatrix},
\end{equation*}
where $c^{\mbox{\tiny{(k)}}}\triangleq \tilde{C}_{\mbox{\tiny{N}}}(\alpha w^{\mbox{\tiny{(k)}}},\alpha\gamma)$. We also define $\textbf{b}\triangleq(C-\alpha\theta, P, 1)^T$. Finally, the inner-most optimization problem in \eqref{63} can be presented as a standard LP problem, in the spirit of \eqref{LP1}, and specifically, the equality constraints can be presented in vector form as
\begin{equation}
\textbf{Ap}=\textbf{b}.
\end{equation}

Since $\textbf{A}$ is of rank three at most, the Fundamental Theorem of Linear Programming (see section 2.4 in \cite{LP}) states that it is enough to confine the search for the optimal solution to one with at most three non-zero entries. Also, due to similar arguments presented in Section \ref{Section3.2}, we can restrict our search to symmetric solutions, namely vectors $\textbf{p}$ that sustain $\Pr\left \{W=w, S=s\right \}=\Pr\left \{W=-w, S=-s\right \}$ for some $w$ and $s$ along the grid. We can also determine that $w$ and $s$ have the same sign, since our detector rejects $\pazocal{H}_0$ when the correlation is high. Combining these two properties (and a slight abuse of notation) makes the form of the optimal distribution $f_{\mbox{\tiny{WS}}}$ rather simple:
\begin{equation}
f'_{\mbox{\tiny{WS}}}(u,v; k) = p\cdot \delta(u+w,v+s)+(1-2p)\cdot \delta(u,v)+p\cdot \delta(u-w,v-s),
\end{equation}
where $0\le p \le \frac{1}{2}$, $w,s>0$ are quantization levels and $\delta(\cdot )$ is the Dirac delta function. Finally, instead of considering $f'_{\mbox{\tiny{WS}}}(u,v; k)$, we can optimize the quantization levels $w$ and $s$ directly. Again, we end up with a numerically manageable problem, which is equivalent to \eqref{63}:
\begin{multline}\label{64}
{E}_{\mbox{\tiny{MD}}}(E_{\mbox{\tiny{0}}}, \lambda, \gamma)=\sup_{\left \{0\le p\le \frac{1}{2},s,w:\; {E}_{\mbox{\tiny{FA}}}\ge E_{\mbox{\tiny{0}}},\; \mathbb{E}[S^2]\le P_{\mbox{\tiny{s}}}\right \}}\big \{2\lambda pws+2\lambda \gamma p s^2-2p\cdot \tilde{C}_{\mbox{\tiny{V}}}(\lambda w+2\lambda\gamma s,-\lambda\gamma) 
\\- (1-2p)\cdot \tilde{C}_{\mbox{\tiny{V}}}(0,-\lambda\gamma))-\lambda\theta\big \}
\end{multline}

\noindent\textbf{Example 4.} The setting is the same as in Example 3, except here we shall perform a joint optimization of the signal levels, the correlator weights and the energy coefficient. After substituting the FA constraint into the MD error exponent in \eqref{64} and letting $\gamma$ and $\lambda$ be optimized, we have
\begin{multline}\label{70}
{E}_{\mbox{\tiny{MD}}}(E_{\mbox{\tiny{0}}})=\max_{\alpha,\lambda, \gamma,w,s\ge 0,\; 0\le p\le \frac{1}{2},\; 2ps^2\le P_{\mbox{\tiny{s}}}} \bigg\{2\lambda pws+2\lambda p\gamma s^2-2p\cdot \tilde{C}_{\mbox{\tiny{V}}}(\lambda w+2\lambda\gamma s,-\lambda\gamma)  \\
- (1-2p)\cdot \tilde{C}_{\mbox{\tiny{V}}}(0,-\lambda\gamma)-\frac{\lambda}{\alpha}\left [E_{\mbox{\tiny{0}}} + 2p\cdot \tilde{C}_{\mbox{\tiny{N}}}\left (\alpha w,\alpha\gamma\right ) + (1-2p)\cdot \tilde{C}_{\mbox{\tiny{N}}}(0,\alpha\gamma)\right ]\bigg\}.
\end{multline}
The solution to \eqref{70}, for $E_{\mbox{\tiny{0}}}\in [0,4]$ and same parameter values as in Example 3 are presented in Figure \ref{fig:Example4}. We compare the characteristics of the detector in \eqref{70} to the optimal detector optimized to match a signal $s_t\in\{-4,0,4\}$, where $s_t=\pm 4$ along half the time.
\begin{figure}[H]
\centering
\includegraphics[scale=0.5]{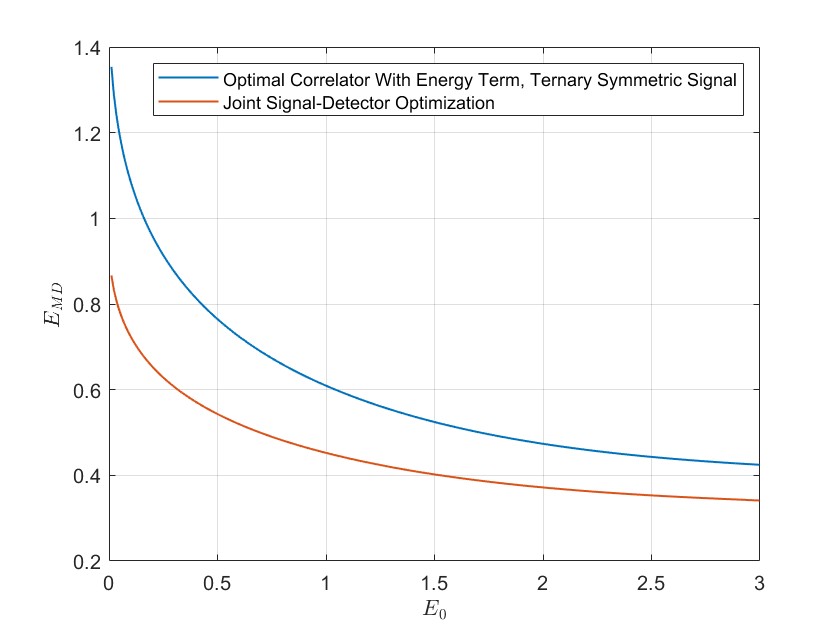}
\caption{MD Error exponents as a function of $E_{\mbox{\tiny{0}}}$ of the optimal correlator (with energy term) and the optimal signal-detector combination (with energy term).}
\label{fig:Example4}
\end{figure}
As expected, joint optimization yielded superior performance.
\section{Conclusion and Future Aspects}
We summarize our findings and propose directions for further exploration.  Initially, we examined the optimal correlator for a given signal within a non-Gaussian noise regime and discovered a non-linear relation between the signal and the correlator weights. Afterwards, by allowing optimization of both $\{s_t\}$ and $\{w_t\}$, we determined that the optimal signal is balanced ternary with $w_t\propto s_t$. Introducing $\gamma\ne 0$ and conducting the joint optimization led us to discover that while $s_t$ and $w_t$ maintain their balanced ternary characteristics, they are not inherently linearly related any more. The resulting optimization problem has 5 positive parameters: the energy coefficient $\gamma$, the non-zero signal level $s$, the corresponding correlator weight $w$, the duty cycle of the signal $p$ and two Chernoff parameters. This reduced parameter set significantly improves the numerical burden of the problem. Our study exhausted the analysis of detectors based on correlation and energy, under the assumptions made regarding the noise distributions.

The first necessary extension to our work is to extend our analysis regarding detectors based on correlation and energy to a broader class of noise distributions. Recall that for the Chernoff bounding technique to be applicable, the CGF of $V^2$ must be finite in an interval around the origin. By using methods related to heavy-tailed large deviations, a broader class of noise distributions can be considered. Another possible extension of our work is to apply the idea of mismatched correlation tests to communication, where each message from a set of size $M>1$ is matched to a waveform $s_i(t)$ along a time interval $0\leq t\leq T$ and transmitted over an additive noisy channel. Each message is transmitted with a prior probability $p_i$ for $1\leq i\leq M$, and all the messages are equal in power, namely, $\frac{1}{T}\int_0^T s_i^2(t) dt = P$ for $1\leq i\leq M$'s . The receiver receives the waveform $X(t)=s_i(t)+N(t)$, and decides which message was transmitted. A classic result in communication theory is that when the noise $N(t)$ is Gaussian, the optimal receiver measures the correlation between $X(t)$ to all $\{s_i(t)\}_{i=1}^M$, and picks the message $i$ that yields the maximal correlation. Assume we insist on using a correlation receiver even when $N(t)$ is not a Gaussian noise process, but some other noise process with a known distribution. We wish to investigate how to set $\textbf{w}$ such that we a obtain minimal error probability within this sub-optimal class. We also note that this problem is intimately related to the subject of mismatched decoding, see \cite{MissmatchedDecoding} and many references therein.
\appendix
\numberwithin{equation}{section}
\section{Appendix}
\begin{lemma}\label{Lemma3}
Let $X$ be a RV with positive variance, such that $\tilde{C}_{\mbox{\tiny{X}}}(a,b)$ is finite and twice-differentiable in some neighborhood of the origin. $\tilde{C}_{\mbox{\tiny{X}}}(a,b)$ is a strictly convex function of $a$ for any fixed $b$.
\begin{proof}
First, assume $b$ is fixed and denote $C(a)=\tilde{C}_{\mbox{\tiny{X}}}(a,b)$. Recall that
\begin{align}
C(a)=\ln\left\{\mathbb{E}\left[e^{aX+bX^2}\right]\right\} = \ln\left\{\int{e^{ax+bx^2}f_{\mbox{\tiny{X}}}(x)\mbox{d}x}\right\}.
\end{align}
The first derivative of $C(a)$ is
\begin{align} \label{17}
\dot{C}(a) = \frac{\int{x\cdot{}e^{ax+bx^2}f_{\mbox{\tiny{X}}}(x)\mbox{d}x}}{\int{e^{ax+bx^2}f_{\mbox{\tiny{X}}}(x)\mbox{d}x}}.
\end{align}
We may define a new PDF, $f(x;a)$,
\begin{align}
f_a(x)\triangleq\frac{e^{ax+bx^2}f_{\mbox{\tiny{X}}}(x)}{\int{e^{ax+bx^2}f_{\mbox{\tiny{X}}}(x)\mbox{d}x}},
\end{align}
that corresponds to $X_a$ which is a non-degenerate RV. Note, that $\dot{C}(a)=\mathbb{E}\left [X_a\right ]$.

We differentiate \eqref{17} again:
\begin{flalign}
\ddot{C}(a)&=\frac{\int{x^2\cdot{}e^{ax+bx^2}f_{\mbox{\tiny{X}}}(x)\mbox{d}x}\cdot\int{e^{ax+bx^2}f_{\mbox{\tiny{X}}}(x)\mbox{d}x}-\int{x\cdot{}e^{ax+bx^2}f_{\mbox{\tiny{X}}}(x)\mbox{d}x}\cdot\int{x\cdot{}e^{ax+bx^2}f_{\mbox{\tiny{X}}}(x)\mbox{d}x}}    {\left(\int{e^{ax+bx^2}f_{\mbox{\tiny{X}}}(x)\mbox{d}x}\right)^2}
& \nonumber \\
&= \frac{\int{x^2\cdot{}e^{ax+bx^2}f_{\mbox{\tiny{X}}}(x)\mbox{d}x}}{\int{e^{ax+bx^2}f_{\mbox{\tiny{X}}}(x)\mbox{d}x}} - \frac{\left(\int{x\cdot{}e^{ax+bx^2}f_{\mbox{\tiny{X}}}(x)\mbox{d}x}\right)^2}{\left(\int{e^{ax+bx^2}f_{\mbox{\tiny{X}}}(x)\mbox{d}x}\right)^2}
& \nonumber \\
&=\mathbb{E}\left[X_a^2\right]-\mathbb{E}\left [X_a\right ]^2
& \nonumber \\
&=\textrm{Var}\{X_a\}>0.
\end{flalign}

\end{proof}
\end{lemma}


\clearpage

\begin{thebibliography}{AA}
\bibitem{BOR09}
F.~Bandiera, D.~Orlando, and G.~Ricci, {\em Advanced Radar Detection Schemes
Under Mismatched Signal Models}, Synthesis Lectures in Signal Processing,
Morgan \& Claypool Publishers, 2009.


\bibitem{GGFL98}
F.~Gini, M.~V.~Greco, A.~Farina, and P.~Lombardo, ``Optimum and mismatched
detection against $K$-distributed plus Gaussian clutter,''
{\em IEEE Trans.\ Aerospace and Electronic Systems}, vol.\ 34, no.\ 3, pp.\
860--876, July 1998.

\bibitem{HLYC20}
C.~Hao, B.~Liu, S.~Yan, and L.~Cai1, ``Parametric adaptive Radar detector with
enhanced mismatched
signals rejection capabilities,''
{\em EURASIP Journal on Advances in Signal Processing},
vol.\ 2010, Article ID 375136, 2010.

\bibitem{LL19}
J.~Liu and J.~Li,
``Robust detection in MIMO Radar with
steering vector mismatches,''
{\em IEEE Trans.\ Signal Processing}, vol.\ 67, no.\ 20, October 2019.

\bibitem{WLYTDY15}
L.~Wei-jian, W.~Li-cai, D.~Yuan-shui, J.~ Tao, X,~Dang, and
W.~Yong-liang, ``Adaptive energy detector and its application for mismatched
signal detection,'' {\em Journal of Radars}, vol.~4, no.~2, pp.\ 149--159, April 2015.

\bibitem{Neri21}
N.~Merhav, ``Optimal Correlators for Detection and Estimation In Optical Receivers'', 
{\em IEEE Trans. Inform. Theory},
vol.~67, no.~8, pp.~5200--5210, August 2021.

\bibitem{Neri22}
N.~Merhav, ``Optimal Correlators and Waveforms for Mismatched Detection'', 
{\em IEEE Trans. Inform. Theory},
vol.~68, no.~5, pp.~8342--8354, December 2022.

\bibitem{Taboga2021}
Taboga, Marco (2021)``Cumulant generating function'', Lectures on probability theory and mathematical statistics. Kindle Direct Publishing. Online appendix. https://www.statlect.com/fundamentals-of-probability/cumulant-generating-function.

\bibitem{LP}
D.G.~Luenberger, Y.~Ye, {\em Linear and Nonlinear Programming}, Chapter 2, Springer International Publishing, Switzerland, 2016.

\bibitem{ChernoffBound}
A.~Dembo, O.~Zeitouni, {\em Large Deviations Techniques and Applications, Second Edition}, Springer-Verlag, New York, 1998.

\bibitem{LargeDeviation}
T.~Mikosch, A.~Nagaev, {\em Large Deviations of Heavy-Tailed Sums with Applications in Insurance. Extremes 1}, Springer, 1998.

\bibitem{MissmatchedDecoding}
J.~Scarlett, A.~Fàbregas, A.~Somekh-Baruch, A.~Martinez, ``Information-Theoretic Foundations of Mismatched Decoding'',
{\em Foundations and Trends in Communications and Information Theory},
vol.~17, no.~2-3,pp.~149--401, 2020.

\end{thebibliography}
\end{document}